\documentclass[draftclsnofoot,onecolumn,12pt]{IEEEtran}
\usepackage{amssymb}
\usepackage{amsmath}
\usepackage{stfloats}
\usepackage{graphicx}
\usepackage{subfigure}
\usepackage{tabularx}
\usepackage{epsfig,epsf,color,balance,cite}
\usepackage{multirow}

\newtheorem{remark}{Remark}
\newtheorem{theorem}{Theorem}

\newtheorem{corollary}{Corollary}
\newtheorem{proof}{Proof}

\begin{document}
%
\title{Tight Capacity Bounds for Indoor Visible Light Communications With Signal-Dependent Noise}

\author{Jin-Yuan Wang,  \IEEEmembership{Member, IEEE},
        Xian-Tao Fu,
        Rong-Rong Lu,\\
        Jun-Bo Wang, \IEEEmembership{Member, IEEE},
        Min Lin,  \IEEEmembership{Member, IEEE},\\
        and Julian Cheng, \IEEEmembership{Senior Member, IEEE}
\thanks{Part of the paper was presented at\emph{ IEEE Global Communications Conference}, Austin, TX, USA \cite{BIB21_1}.}
\thanks{This work was supported in part by the National Natural Science Foundation of China under Grant 61701254, the Open Project of Shanghai Key Laboratory of Trustworthy Computing, the National Key R\&D Program of China under Grants 2018YFB2202200 \& 2020YFB1807803, and the Key International Cooperation Research Project under Grant 61720106003.}
\thanks{Jin-Yuan Wang is with Key Laboratory of Broadband Wireless Communication and Sensor Network Technology, Nanjing University of Posts and Telecommunications, Nanjing 210003, China, and also with Shanghai Key Laboratory of Trustworthy Computing, East China Normal University, Shanghai 200062, China (corresponding author, e-mail: jywang@njupt.edu.cn).}
\thanks{Xian-Tao Fu, Rong-Rong Lu, and Min Lin are with Key Laboratory of Broadband Wireless Communication and Sensor Network Technology, Nanjing University of Posts and Telecommunications, Nanjing 210003, China (e-mail: 1019010206@njupt.edu.cn, 1019010205@njupt.edu.cn, linmin@njupt.edu.cn).}
\thanks{Jun-Bo Wang is with National Mobile Communications Research Laboratory, Southeast University, Nanjing 211111, China (e-mail: jbwang@seu.edu.cn).}
\thanks{Julian Cheng is with School of Engineering, The University of British Columbia, Kelowna, BC, Canada V1V 1V7 (e-mail: julian.cheng@ubc.ca).}
}

\maketitle \linespread{1.45}
\begin{abstract}
Channel capacity bounds are derived for a point-to-point indoor visible light communications (VLC) system with signal-dependent Gaussian noise.
Considering both illumination and communication, the non-negative input of VLC is constrained by peak and average optical intensity constraints.
Two scenarios are taken into account: one scenario has both average and peak optical intensity constraints,
and the other scenario has only average optical intensity constraint.
For both two scenarios,
we derive closed-from expressions of capacity lower and upper bounds.
Specifically, the capacity lower bound is derived by using the variational method and the property that the output entropy is invariably larger than the input entropy.
The capacity upper bound is obtained by utilizing the dual expression of capacity and the principle of ``capacity-achieving source distributions that escape to infinity".
Moreover, the asymptotic analysis shows that the asymptotic performance gap between the capacity lower and upper bounds approaches zero.
Finally, all derived capacity bounds are confirmed using numerical results.
\end{abstract}


\IEEEpeerreviewmaketitle

\newpage
\baselineskip=8.5mm

\section{Introduction}
\label{section1}
The idea of using light as a communication medium was first proposed by Bell in 1880 \cite{BIB01}.
Due to technical limitations, such a communication method has not drawn much attention for a long period.
With the fast development of light-emitting diodes (LEDs),
visible light communications (VLC) has attracted considerable interest recently \cite{BIB02,BIB03}.
As a compelling alternative technology for radio frequency (RF) communications,
VLC is free from RF interference and is license-free \cite{BIB04}.
Undoubtedly, VLC has become a promising candidate for future wireless communications.

When evaluating a communication system,
channel capacity is an important performance metric.
While several works have analyzed the capacity for free-space optical (FSO) communication,
owing to some unique features of VLC, these theoretical results derived for FSO channels cannot be directly applied to indoor VLC.
First, indoor VLC can simultaneously provide communication and illumination \cite{BIB15}, but FSO communication is only used for communication.
Second, due to the indoor illumination requirement, the brightness of light in VLC does not fluctuate with time, but it can be adjusted according to dimming requirement \cite{BIB16}. Therefore, the average optical intensity in VLC is constrained to be equal to a user-defined brightness target.
However, the average optical intensity in FSO communication must be less than a specific safety level to avoid eye and skin injury.
In other words, a lower average optical intensity is often preferred for FSO communication.
Moreover, due to the luminous capability of LED, the peak optical intensity in VLC should be less than an allowed threshold. Consequently, the average and peak constraints on optical intensity must be considered in practical VLC systems.
Furthermore, the Gaussian noises in VLC are often assumed to be independent of the input signal.
Such an assumption is reasonable when the ambient light or thermal noise is large.
Nevertheless, typical illumination environments in VLC result in large signal-to-noise ratio (SNR) \cite{BIB20,BIB21}.
In the high SNR regime, such a signal-noise-independent assumption ignores a basic issue:
because the LED in VLC can randomly radiate photons,  the value of noise relies on the input signal \cite{BIB21_1}, which suggests that the signal-dependent noise should also be considered for indoor VLC.
However, the input constraints and the signal-dependent noise impose tremendous challenges in calculating the channel capacity of VLC.

In this paper, the capacity of an indoor VLC system will be investigated.
Because finding an exact closed-form capacity expression is challenging, this paper derives tight capacity bounds.

\subsection{Related Works}
In FSO communication, Poisson and Gaussian channels are often adopted for capacity analysis.
Under Poisson channels, the capacities of FSO communication were analyzed \cite{BIB05,BIB06,BIB07}.
Over the more widely used Gaussian channels, the capacity bounds of FSO communication were derived \cite{BIB08,BIB09,BIB10}.
Based on the sphere-packing and the maximum entropy distribution,
 the capacity upper and lower bounds were obtained \cite{BIB11}.
In \cite{BIB12}, the capacities over the noncoherent and partially coherent Gaussian channels were derived.
In \cite{BIB13}, the capacity was evaluated for FSO communication using optical orthogonal frequency division multiplexing (OFDM).
In \cite{BIB14}, a novel upper bound of the capacity was obtained by using the method of trigonometric moment space.

Although many works have investigated the capacity of FSO communication, the analytical expressions derived for FSO channels cannot be applied to VLC directly.
Recently, the capacity analysis of VLC has drawn much attention.
In \cite{BIB16}, the capacity for VLC was analyzed, but capacity expression was not derived.
In \cite{BIB17}, tight capacity bounds were obtained for indoor VLC having non-negativity and average optical intensity constraint.
With an additional peak optical intensity constraint,
the capacity of VLC was further investigated \cite{BIB18}.
Refs. \cite{BIB19} and \cite{BIB19_1} investigated the capacities for indoor VLC using pulse amplitude modulation and OFDM. In most works, the Gaussian noises \cite{BIB16,BIB17,BIB18,BIB19,BIB19_1} are assumed to be independent of the signal.
By considering the signal-dependent noise for FSO communication,
the channel capacity \cite{BIB22}, secrecy capacity \cite{add00}, and pre-coding scheme \cite{add01} were obtained.
For VLC with signal-dependent noise, efficient modulation schemes \cite{add1} and transceiver design methods \cite{add2} were studied.
In our preliminary study \cite{BIB21_1}, the capacity bounds for VLC with signal-dependent noise was derived.
However, the asymptotic analysis and the impacts of peak optical intensity on capacity performance were not provided.

\subsection{Contributions and Organization}
In this paper, we consider a point-to-point indoor VLC system with signal-dependent Gaussian noise and analyze the channel capacity of such a system.
The main contributions of this paper are summarized as follows:
\begin{itemize}
  \item Under the average and peak optical intensity constraints, we investigate the capacity bounds for indoor VLC. A capacity lower bound is obtained based on the property that the output entropy is larger than the input entropy. An optimal input probability density function (PDF) is derived by a variational method.
      By using the dual expression of capacity and the principle of ``capacity-achieving source distributions that escape to infinity", a capacity upper bound is obtained. Numerical results verify the accuracy of the derived bounds.
  \item By removing the peak optical intensity constraint, we further analyze the capacity bounds for indoor VLC. For this scenario, a novel input PDF is derived by solving an optimization problem. Moreover, closed-form capacity lower and upper bounds in this case are also derived. The accuracy of the derived bounds is verified.  The capacity bounds for VLC having different constraints are also compared.
  \item The tightness of the derived capacity bounds is examined, and the asymptotic analysis is also provided. For VLC having different constraints, the performance gap and the asymptotic performance gap between the capacity upper and lower bounds are derived. At hight optical intensity, we show that the gap between the asymptotic upper and lower bounds approaches zero. This observation indicates that the derived capacity upper and lower bounds are asymptotically tight.
\end{itemize}

The remainder of the paper is organized as follows.
Section \ref{section2} introduces the system model.
Section \ref{section3} analyzes channel capacity bounds for VLC having both average and peak optical intensity constraints.
Section \ref{section4} derives channel capacity bounds for VLC having only average optical intensity constraint.
Numerical results are presented in Section \ref{section5} and conclusions are drawn in Section \ref{section6}.

\emph{Notations:} In this paper, ${\cal N} (\mu,\sigma^2)$ denotes a Gaussian distribution having mean $\mu$ and variance $\sigma^2$;
$f_X(\cdot)$ and $f_{Y|X}(\cdot)$ are, respectively, the PDF of $X$ and the conditional PDF of $Y$ given $X$;
$E(\cdot)$ is the expectation operator;
$I(\cdot;\cdot)$ is the mutual information;
$H(\cdot)$ and $H( \left. \cdot \right|\cdot )$ are the entropy and conditional entropy;
${\rm{erf}}(\cdot) $ is the error function;
$D\left( {\left.  \cdot  \right\| \cdot } \right)$ is the relative entropy;
${\cal Q}(\cdot) $ is the Gaussian ${\cal Q}$-function;
${o_X}(1)$ is an infinitely small quantity when $X$ approaches infinity.

\section{System Model}
\label{section2}
As depicted in Fig. \ref{fig1}, we consider a VLC system consisting of a transmitter and a receiver.
At the transmitter, the light source is an LED,
which performs the electrical-to-optical conversion.
The optical signal is propagated via the VLC channel.
At the receiver, a photodiode is utilized to convert the optical signal to an electrical signal.
In the system, thermal noise, shot noise and amplifier noise at the receiver are considered \cite{BIB23}.
Specifically, both thermal noise and amplifier noise are signal-independent and are assumed to follow Gaussian distributions \cite{BIB10}.
The shot noise can also be modeled as a Gaussian distribution, but its strength relies on the input signal \cite{BIB23}.
Consequently, the received signal $Y$ at the receiver side is modeled as\footnote{Following \cite{BIB10} and \cite{BIB22}, we consider a discrete-time channel. All entropies in the ensuing derivation are presented in nats, as natural logarithm is employed. Therefore, the channel capacity in this paper is presented in nats/channel use \cite{BIB25}.}
\begin{equation}
Y = X + \sqrt {X} {Z_1} + {Z_0},
\label{equ1}
\end{equation}
where $X$ is the optical intensity;
${Z_0} \sim {\cal N}(0,{\sigma ^2})$ is the signal-independent additive white Gaussian noise (AWGN);
$\sqrt{X} Z_1$ is the signal-dependent AWGN where $Z_1 \sim {\cal N} (0, {\varsigma ^2} {\sigma ^2})$, and the term ${\varsigma ^2} > 0$ denotes a scaling factor typically ranging from 0 to 10 \cite{add1};
Both noise terms ${Z_0}$ and $\sqrt{X}{Z_1}$ are assumed to be independent.
Note that, in (\ref{equ1}), the losses and optoelectronic conversion factor in VLC are constants and scale the SNR only. Appendix \ref{Appendix_A} justifies the model in (\ref{equ1}).

\begin{figure}
\centering
\includegraphics[width=10cm]{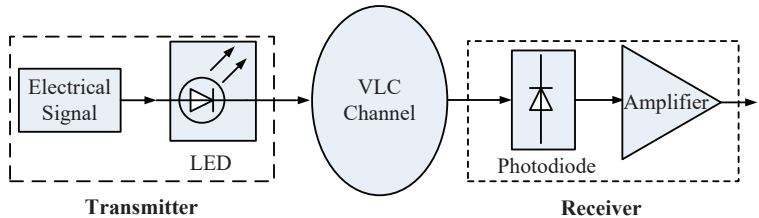}
\caption{A typical indoor VLC system}
\label{fig1}
\end{figure}

In indoor VLC, the following three constraints on $X$ should be considered:
\begin{description}
  \item[1)] \emph{Non-negativity}: For indoor VLC, the optical intensity is typically employed to carry information \cite{BIB24}, and thus the transmitted signal $X$ must be non-negative, i.e.,
\begin{eqnarray}
X \ge 0. \label{equ3}
\end{eqnarray}
  \item[2)] \emph{Peak optical intensity constraint}: For practical indoor VLC systems, the peak optical intensity is constrained by the luminous capability of the LED \cite{BIB24}. This means that $X$ cannot exceed the LED's peak optical intensity $A$, i.e.,
\begin{eqnarray}
X \le A. \label{equ4}
\end{eqnarray}
  \item[3)] \emph{Average optical intensity constraint}: According to the illumination requirement in IEEE 802.15.7,
the average optical intensity in indoor environments does not change with time,
but it can be adjusted according to the dimming requirement \cite{BIB18}, i.e.,
\begin{eqnarray}
 E_X(X) = \int_0^A xf_X(x) {\rm d}x  = \xi P, \label{equ5}
\end{eqnarray}
where $P$ is the LED's nominal optical intensity, and $\xi \in [0,1]$ is the dimming target.
\end{description}

Define an average-to-peak-optical-intensity ratio (APR) as
\begin{eqnarray}
  \alpha \buildrel \Delta \over = \xi P/A. \label{add1}
\end{eqnarray}
In this paper, we consider two indoor VLC scenarios. For the scenario having both average and peak optical intensity constraints, we have $0<\alpha \leq 1$. For the scenario having only an average optical intensity constraint, we let $\alpha \to 0$.

For the channel model in (\ref{equ1}) with constraints (\ref{equ3})-(\ref{equ5}),
it is challenging to find an exact capacity expression.
Alternatively, tight bounds on capacity will be investigated.

\section{Channel Capacity With Both Average and Peak Optical Intensity Constraints}
\label{section3}
In this section, by considering constraints (\ref{equ3})-(\ref{equ5}),
we derive the lower and upper bounds of the channel capacity for VLC.
Then, the tightness of the derived bounds is investigated.

\subsection{Capacity Lower Bound}
\label{section3_1}
According to information theory, a capacity lower bound can be derived by calculating the mutual information under an arbitrary input PDF.
Consequently, the capacity of VLC can be lower-bounded by
\begin{eqnarray}
 C   &\ge& {\left. {I\left( {X;Y} \right)} \right|_{{\rm{any}}\;f_X(x)\;{\rm{satisfing}}\;{(\ref{equ3}),}\;{(\ref{equ4})}\;{\rm{and}}\;{(\ref{equ5})}}}\nonumber  \\
  &=& H(Y)-H\left( {\left. Y \right|X} \right).
\label{equ9}
\end{eqnarray}

According to (\ref{equ1}), the conditional PDF $f_{Y|X}(y|x)$ is given by
\begin{equation}
{f_{Y\left| X \right.}}\left( {y\left| x \right.} \right) = \frac{1}{{\sqrt {2\pi \left( {1 + x{\varsigma ^2}} \right){\sigma ^2}} }}{e^{ - \frac{{{{\left( {y - x} \right)}^2}}}{{2\left( {1 + x{\varsigma ^2}} \right){\sigma ^2}}}}},
\label{equ2}
\end{equation}
and thus the conditional entropy $H(Y|X)$ in (\ref{equ9}) is expressed as
\begin{eqnarray}
H\left( {Y\left| X \right.} \right)  = \frac{1}{2}\ln \left( {2\pi e{\sigma ^2}} \right) + \frac{1}{2}{E_X} ( {\ln ( 1 + X{\varsigma ^2} )} ).
\label{equ10}
\end{eqnarray}

According to \emph{Proposition 11} in \cite{BIB22},
it is known that the output entropy $H(Y)$ is invariably larger than the input entropy $H(X)$, i.e.,
\begin{equation}
H(Y) \ge H(X) + {f_{{\rm{low}}}}\left( {\xi P} \right),
\label{equ11}
\end{equation}
where ${f_{{\rm{low}}}}\left( {\xi P} \right)$ is positive and defined as
\begin{equation}
{f_{{\rm{low}}}}\left( {\xi P} \right) = \frac{1}{2}\ln \left( {1 + \frac{{2{\varsigma ^2}{\sigma ^2}}}{{\xi P}}} \right) - \frac{{\xi P + {\varsigma ^2}{\sigma ^2}}}{{{\varsigma ^2}{\sigma ^2}}} + \frac{{\sqrt {\xi P\left( {\xi P + 2{\varsigma ^2}{\sigma ^2}} \right)} }}{{{\varsigma ^2}{\sigma ^2}}}.
\label{equ12}
\end{equation}

Substituting (\ref{equ10}) and (\ref{equ11}) into (\ref{equ9}), we have
\begin{eqnarray}
C \!\!\!\!&\ge&\!\!\!\! -{\cal J}\left[ f_X(x) \right]+ {f_{{\rm{low}}}}\left( {\xi P} \right)- \frac{1}{2}\ln \left( {2\pi e{\sigma ^2}} \right),
\label{equ13}
\end{eqnarray}
where ${\cal J}\left[ f_X(x) \right]$ is given by
\begin{equation}
{\cal J}\left[ f_X(x) \right] = {  \int_0^A {f_X(x)} \ln ( {f_X(x)} ) {\rm{d}}x +\frac{1}{2}\int_0^A {\ln \left( {1 + {\varsigma ^2}x} \right)f_X(x)} {\rm{d}}x}.
\label{equ:23}
\end{equation}
In (\ref{equ13}), a lower bound can be obtained by choosing an arbitrary $f_X(x)$ satisfying (\ref{equ3}), (\ref{equ4}) and (\ref{equ5}).
To obtain a tight lower bound on capacity, we must select a suitable $f_X(x)$.
Here, we consider the following optimization problem
\begin{eqnarray}
  \begin{split}
    & \min_{f_X(x)} {\cal J}\left[ f_X(x) \right] \\
    \text{s.t.}\qquad & \int_0^A {f_X(x){\rm{d}}x}  = 1  \\
    & \int_0^A xf_X(x){\rm{d}}x  = \xi P.
  \end{split}
  \label{equ14}
\end{eqnarray}
Note that problem (\ref{equ14}) can be solved via the classic variational method \cite{BIB25_1}.
Then, we can obtain the following theorem.

\begin{theorem}
The optimal PDF $f_X(x)$ for the optimization problem (\ref{equ14}) is given by\footnote{For FSO channels under average and peak optical intensity constraints, the capacity-achieving input PDF is shown to be discrete with a finite number of mass points \cite{BIB08,BIB09}. However, the explicit PDF expression is not available. No clear insights can be drawn for system design.
Here, when analyzing the capacity lower bound of a VLC channel, we provide a continuous and explicit expression for the input PDF. Based on the tightness analysis in Section \ref{section3_3}, we will show that the derived input PDF in \emph{Theorem \ref{the1}} is capacity-approaching.
}
\begin{eqnarray}
{f_X}(x) = \left\{ \begin{array}{l}
\frac{{{\varsigma ^2}}}{{2\left( {\sqrt {1 + {\varsigma ^2}A}  - 1} \right)\sqrt {1 + {\varsigma ^2}x} }},\;x \in [0,A],\;{\rm{if}}\;\alpha  = \frac{{{\varsigma ^2}A + \sqrt {1 + {\varsigma ^2}A}  - 1}}{{3{\varsigma ^2}A}}\\
\frac{{{\varsigma ^2}{e^{bx}}}}{{2g(b,{\varsigma ^2},A)\sqrt {1 + {\varsigma ^2}x} }},\;x \in [0,A],\;\;\;\;\;\;\;{\rm{if}}\;\alpha  \ne \frac{{{\varsigma ^2}A + \sqrt {1 + {\varsigma ^2}A}  - 1}}{{3{\varsigma ^2}A}}\;{\rm{and}}\;\alpha  \in (0,P/A]
\end{array} \right.,
\label{b1}
\end{eqnarray}
where $\alpha$ is given by (\ref{add1}), $g\left( {b,{\varsigma ^2},A} \right)$ is defined as
\begin{eqnarray}
g\left( {b,{\varsigma ^2},A} \right) \buildrel \Delta \over = \left\{ {\begin{array}{*{20}{c}}
{\frac{{{e^{ - \frac{b}{{{\varsigma ^2}}}}}\varsigma \sqrt \pi  }}{{2\sqrt { - b} }}\left[ {{\rm{erf}}\left( {\sqrt {-\frac{{  b\left( {1 + {\varsigma ^2}A} \right)}}{{{\varsigma ^2}}}} } \right) - {\rm{erf}}\left({\sqrt {- \frac{{ b}}{{{\varsigma ^2}}}} } \right)} \right]  ,b < 0}\\
{\int_1^{\sqrt {1 + {\varsigma ^2}A} } {{e^{\frac{{b\left( {{t^2} - 1} \right)}}{{{\varsigma ^2}}}}}{\rm{d}}t},\;\;\;\;\;\;\;\;\;\;\;\;\;\;\;\;\;\;\;\;\;\;\;\;\;\;\;\;\;\;\;\;\;\;\;\;\;\;\;\;b > 0}
\end{array}} \right.,
\label{equ19}
\end{eqnarray}
and $b$ is the solution of the following transcedental equation
\begin{eqnarray}
\xi P = \frac{{\sqrt {1 + {\varsigma ^2}A} {e^{bA}} - 1}}{{2bg\left( {b,{\varsigma ^2},A} \right)}} - \frac{1}{{2b}} - \frac{1}{{{\varsigma ^2}}}.
\label{equ20}
\end{eqnarray}
\label{the1}
\end{theorem}

\begin{proof}
See Appendix \ref{Appendix_B}.
\end{proof}

According to (\ref{equ13}) and \emph{Theorem \ref{the1}}, the following theorem is obtained.

\begin{theorem}
For indoor VLC with constraints (\ref{equ3}), (\ref{equ4}) and (\ref{equ5}), a lower bound of the channel capacity is derived as\footnote{The lower bound in \emph{Theorem \ref{the2}} is valid for all values of optical intensities.}
\begin{eqnarray}
C_{\rm low} = \left\{ \begin{array}{l}
\ln \left( {\frac{2\left( {\sqrt {1 + {\varsigma ^2}A}  - 1} \right)}{\varsigma ^2 \sqrt{2 \pi e {\sigma^2}}}} \right) + {f_{{\rm{low}}}}\left( {\xi P} \right),\;{\rm{if}}\;\alpha  = \frac{{{\varsigma ^2}A + \sqrt {1 + {\varsigma ^2}A}  - 1}}{{3{\varsigma ^2}A}}\\
\!\ln\! \left( {\frac{{2g\left( {b,{\varsigma ^2},A} \right)}}{\varsigma ^2 \sqrt{2 \pi e {\sigma^2}}}} \right) - b\xi P+ {f_{{\rm{low}}}}\left( {\xi P} \right),\;{\rm{if}}\;\alpha  \ne \frac{{{\varsigma ^2}A + \sqrt {1 + {\varsigma ^2}A}  - 1}}{{3{\varsigma ^2}A}}\;{\rm{and}}\;\alpha  \in (0,P/A]
\end{array} \right.,
\label{a1}
\end{eqnarray}
where $f_{\rm low}(\xi P)$ and $g\left( {b,{\varsigma ^2},A} \right)$ are defined in (\ref{equ12}) and (\ref{equ19}), respectively.
\label{the2}
\end{theorem}
\begin{proof}
See Appendix \ref{Appendix_C}.
\end{proof}

\subsection{Capacity Upper Bound}
\label{section3_2}
In this subsection, the upper bound of capacity is derived by using the dual expression of the capacity.
According to \cite{BIB26,BIB27,BIB28}, the following identity holds
\begin{eqnarray}
E_X( D(f_{Y|X}(y|X)\|R_Y(y)) )= I(X;Y) + D(f_{Y}(y) \| R_Y(y)),
\label{equ23}
\end{eqnarray}
where $R_Y(y)$ stands for an arbitrary PDF on $Y$. According to \emph{Theorem 2.6.3} in \cite{BIB25}, we know $D( f_Y(y)\| R_Y(y)) \ge 0$. Therefore, eq. (\ref{equ23}) can be reduced to
\begin{eqnarray}
 I\left( {X;Y} \right) \le E_X( D(f_{Y|X}(y|X)\|R_Y(y)) ).
\label{equ25}
\end{eqnarray}
Then, the capacity of VLC can be upper-bounded by
\begin{eqnarray}
 C \le  {E_{{X:f_X(x)=f^*_X(x)}}}( D(f_{Y|X}(y|X)\|R_Y(y)) ),
\label{equ26}
\end{eqnarray}
where $f^*_X(x)$ denotes the channel capacity-achieving input PDF.
Moreover, the relative entropy $D(f_{Y|X}(y|X)\|R_Y(y))$ can be further expressed as
\begin{eqnarray}
D\left( {{f_{Y|X}}\left( {y\left| X \right.} \right)\left\| {{R_Y}\left( y \right)} \right.} \right)
 =  - \frac{1}{2}\ln \left( {2\pi e{\sigma ^2}\left( {1 + {\varsigma ^2}X} \right)} \right)\!\! -\!\! \int_{ - \infty }^{ \infty } {{f_{Y|X}}\left( {y\left| X \right.} \right)\ln \left( {{R_Y}\left( y \right)} \right)} \,{\rm{d}}y.
\label{equ27}
\end{eqnarray}

In (\ref{equ26}), an upper bound of capacity can be obtained by selecting an arbitrary $R_Y(y)$.
However, a suitable choice of $R_Y(y)$ is required to obtain a tight bound.
Here, when $\alpha  = ({\varsigma ^2}A + \sqrt {1 + {\varsigma ^2}A}  - 1 ) / ( 3{\varsigma ^2}A )$, $R_Y(y)$ can be chosen as
\begin{eqnarray}
{R_Y}\left( y \right) = \left\{ \begin{array}{c}
\frac{{2\beta }}{{\sqrt {2\pi } }}{e^{ - \frac{{{y^2}}}{2}}}, \quad \quad \quad \quad \quad \;\;y \in \left( { - \infty ,0} \right)\\
\frac{{\left( {1 - 2\beta } \right){\varsigma ^2}}}{{2\left( {\sqrt {1 + {\varsigma ^2}\left( {A + A\delta } \right)}  - 1} \right)\sqrt {1 + {\varsigma ^2}y} }},  \;y \in \left[ {0,A + A\delta } \right]\\
\frac{\beta }{{{e^{ - \left( {A + A\delta } \right)}}}}{e^{ - y}}, \;\;\;\;\;\quad\quad\quad \quad \quad y \in \left( {A + A\delta ,\infty } \right)
\end{array} \right..
\label{equ28}
\end{eqnarray}
When $\alpha  \ne ({\varsigma ^2}A + \sqrt {1 + {\varsigma ^2}A}  - 1 ) / ( 3{\varsigma ^2}A )$ and $\alpha  \in ( 0, P / A] $, $R_Y(y)$ is selected as
\begin{eqnarray}
{R_Y}\left( y \right) = \left\{ {\begin{array}{*{20}{c}}
{\frac{{2\beta }}{{\sqrt {2\pi } }}{e^{ - \frac{{{y^2}}}{2}}}, \quad  \quad\; y \in \left( { - \infty ,0} \right)}\\
\begin{array}{l}
\frac{{\left( {1 - 2\beta } \right){\varsigma ^2}{e^{by}}}}{{2G\left( {b,{\varsigma ^2},A,\delta } \right)\sqrt {1 + {\varsigma ^2}y} }}, \quad  y \in \left[ {0,A + A\delta } \right]\\
\frac{\beta }{{{e^{ - \left( {A + A\delta } \right)}}}}{e^{ - y}},  \quad \quad \quad \;y \in \left( {A + A\delta ,\infty } \right)
\end{array}
\end{array}} \right.,
\label{equ29}
\end{eqnarray}
where $\beta \in (0,1) $ and $\delta >0$ are small positive constants, and $G\left( {b,{\varsigma ^2},A,\delta } \right)$ is defined as
\begin{equation}
G\left( {b,{\varsigma ^2},A,\delta } \right) \triangleq \int_1^{\sqrt {1 + {\varsigma ^2}A\left( {1 + \delta } \right)} } {{e^{b\frac{{{t^2} - 1}}{{{\varsigma ^2}}}}}{\rm{d}}t}. \label{add2}
\end{equation}

Using (\ref{equ28}), (\ref{equ29}), and the concept of ``capacity-achieving source distributions that escape to infinity" \cite{BIB22,BIB28,BIB29}, we derive the upper bound on capacity in the following theorem.

\begin{theorem}
For indoor VLC having constraints (\ref{equ3}), (\ref{equ4}) and (\ref{equ5}), an upper bound on the capacity is obtained as\footnote{Because the principle of ``capacity-achieving source distributions that escape to infinity" is employed, the upper bound in \emph{Theorem \ref{the3}} is an asymptotic one. This means that the upper bound is valid only asymptotically, i.e., only when the optical intensity is sufficiently large.}
\begin{eqnarray}
C_{\rm upp} = \left\{ \begin{array}{l}
\ln \left( {\frac{{2(\sqrt {1 + {\varsigma ^2}A(1+\delta)}  - 1)}}{{(1 - 2\beta ){\varsigma ^2}} \sqrt{2 \pi e {\sigma^2}}}} \right)  + {o_A}(1),\;{\rm{if}}\;\alpha  = \frac{{{\varsigma ^2}A + \sqrt {1 + {\varsigma ^2}A}  - 1}}{{3{\varsigma ^2}A}}\\
\ln \left( {\frac{{2G(b,{\varsigma ^2},A,\delta )}}{{(1 - 2\beta ){\varsigma ^2}} \sqrt{2 \pi e {\sigma^2}}}} \right) + \psi \left( {b,{\varsigma ^2},A,{\sigma ^2},\xi ,P} \right) + {o_A}(1),\\
\;\;\;\;\;\;\;\;\;\;\;\;\;\;\;\;\;\;\qquad\qquad\qquad\qquad {\rm{if}}\;\alpha  \ne \frac{{{\varsigma ^2}A + \sqrt {1 + {\varsigma ^2}A}  - 1}}{{3{\varsigma ^2}A}}\;{\rm{and}}\;\alpha  \in (0,P/A]
\end{array} \right.,
\label{a2}
\end{eqnarray}
where $\beta \in (0,1) $, $\delta >0$, $G\left( {b,{\varsigma ^2},A,\delta } \right)$ is given by (\ref{add2}), and $\psi \left( {b,{\varsigma ^2},A,{\sigma ^2},\xi ,P} \right)$ is defined as
\begin{eqnarray}
\psi \!( {b,{\varsigma ^2}\!,A,{\sigma ^2}\!,\xi\! ,\!P} ) \!\!\buildrel \Delta \over = \!\!\left\{\!\!\! {\begin{array}{*{20}{c}}
{ - b\frac{{\sqrt {\left( {1 + A{\varsigma ^2}} \right){\sigma ^2}} }}{{\sqrt {2\pi } }}{e^{ - \frac{{{A^2}}}{{2\left( {1 + A{\varsigma ^2}} \right){\sigma ^2}}}}} - b\xi P,\;\;\;\qquad\qquad\qquad\qquad\qquad b < 0}\\
{\!b\frac{{\sqrt {\left( {1 \!+\! A{\varsigma ^2}} \right){\sigma ^2}} }}{{\sqrt {2\pi } }}{e^{ - \frac{{{{\left( {A\delta } \right)}^2}}}{{2\left( {1 \!+\! A{\varsigma ^2}} \right){\sigma ^2}}}}}\!\! -\!\! b\xi P\!\!\!\left[ {{\cal Q}\!\!\!\left( \!\!{\frac{{ - \xi P}}{{\sqrt {\left( {1 \!+\! \xi P{\varsigma ^2}} \right){\sigma ^2}} }}} \!\!\right)\!\! -\!\! {\cal Q}\!\!\!\left( \!\!{\frac{{A + A\delta  - \xi P}}{{\sqrt {\left( {1 \!+\! \xi P{\varsigma ^2}} \right){\sigma ^2}} }}} \right)}\!\! \right]\!,\! b > 0}
\end{array}} \right..
\label{equ33}
\end{eqnarray}
\label{the3}
\end{theorem}

\begin{proof}
See Appendix \ref{Appendix_D}.
\end{proof}

\subsection{Tightness of the Derived Bounds}
\label{section3_3}
In this subsection, the tightness of the derived channel capacity bounds will be investigated.
Moreover, for large $A$, the asymptotic tightness performance will also be analyzed.

According to \emph{Theorems \ref{the2}} and \emph{\ref{the3}}, the gap between upper and lower bounds is obtained in the following theorem.

\begin{theorem}
For indoor VLC with constraints (\ref{equ3}), (\ref{equ4}) and (\ref{equ5}), the performance gap between upper bound (\ref{a2}) and lower bound  (\ref{a1}) is given by
\begin{eqnarray}
{C_{{\rm{gap}}}} \!\!\!\!&=&\!\!\!\! {C_{{\rm{upp}}}} - {C_{{\rm{low}}}} \nonumber \\
 &=&\!\!\!\! \left\{ \begin{array}{l}
\ln \left( {\frac{{\sqrt {1 + {\varsigma ^2}A(1 + \delta )}  - 1}}{{(1 - 2\beta )(\sqrt {1 + {\varsigma ^2}A}  - 1)}}} \right) + {o_A}(1) - {f_{{\rm{low}}}}\left( {\xi P} \right),\;{\rm{if}}\;\alpha  = \frac{{{\varsigma ^2}A + \sqrt {1 + {\varsigma ^2}A}  - 1}}{{3{\varsigma ^2}A}}\\
\ln \left( {\frac{{G(b,{\varsigma ^2},A,\delta )}}{{(1 - 2\beta )g\left( {b,{\varsigma ^2},A} \right)}}} \right) + \psi \left( {b,{\varsigma ^2},A,{\sigma ^2},\xi ,P} \right) + {o_A}(1) + b\xi P - {f_{{\rm{low}}}}\left( {\xi P} \right),\\
\;\;\;\;\;\qquad\qquad\qquad\qquad\qquad\qquad{\rm{if}}\;\alpha  \ne \frac{{{\varsigma ^2}A + \sqrt {1 + {\varsigma ^2}A}  - 1}}{{3{\varsigma ^2}A}}\;{\rm{and}}\;\alpha  \in (0,P/A]
\end{array}\right..
\end{eqnarray}
\label{the4}
\end{theorem}

\begin{proof}
The proof is straightforward and is omitted.
\end{proof}

In an indoor VLC environment, the peak optical intensity $A$ is often large to meet the illumination requirements.
Therefore, we are more interested in the asymptotic performance at large peak optical intensity.
Based on \emph{Theorem \ref{the4}}, the following corollary is derived.

\begin{corollary}
\label{cor1}
For VLC with constraints (\ref{equ3}), (\ref{equ4}) and (\ref{equ5}), the asymptotic performance gap between upper bound (\ref{a2}) and lower bound  (\ref{a1}) is given by
\begin{eqnarray}
\mathop {\lim }\limits_{A \to \infty } {C_{{\rm{gap}}}} = \left\{ \begin{array}{l}
\ln \left( {\frac{{\sqrt {1 + \delta } }}{{1 - 2\beta }}} \right),\;{\rm{if}}\;\alpha  = \frac{{{\varsigma ^2}A + \sqrt {1 + {\varsigma ^2}A}  - 1}}{{3{\varsigma ^2}A}}\\
\ln \left( {\frac{1}{{1 - 2\beta }}} \right), \;{\rm{if}}\;\alpha  \ne \frac{{{\varsigma ^2}A + \sqrt {1 + {\varsigma ^2}A}  - 1}}{{3{\varsigma ^2}A}}\;{\rm{and}}\;\alpha  \in (0,P/A]
\end{array} \right..
\end{eqnarray}
\label{coro1}
\end{corollary}

\begin{proof}
See Appendix \ref{appd}.
\end{proof}

\begin{remark}
Because both $\delta$ and $\beta$ are small positive numbers,
the asymptotic performance gap in \emph{Corollary \ref{coro1}} is small and can be ignored.
This indicates the derived upper and lower bounds are tight for large $A$.
\label{rem1}
\end{remark}

\section{Channel Capacity With Only The Average Optical Intensity Constraint}
\label{section4}
In this section, we investigate the channel capacity bounds without peak optical intensity constraint.
Based on constraints (\ref{equ3}) and (\ref{equ5}), the capacity bounds will be derived.
Moreover, the tightness of the derived bounds will also be verified.

\subsection{Capacity Lower Bound}
\label{section4_1}
For an indoor VLC system having constraints (\ref{equ3}) and (\ref{equ5}),
a lower bound of capacity can also be expressed as (\ref{equ13}).
In this scenario, ${\cal J}\left[ f_X(x) \right]$ becomes
\begin{equation}
{\cal J}\left[ f_X(x) \right] = {  \int_0^\infty {f_X(x)} \ln \left( {f_X(x)} \right){\rm{d}}x +\frac{1}{2}\int_0^\infty {\ln \left( {1 + {\varsigma ^2}x} \right)f_X(x)} {\rm{d}}x}.
\end{equation}
To obtain a tight lower bound, we formulate the following optimization problem
\begin{eqnarray}
\begin{split}
&\mathop {\min }\limits_{{f_X}\left( x \right)} \; {\cal J}[f_X(x)]  \\
{\rm s.t.}\qquad & {\int_0^\infty  {{f_X}\left( x \right){\rm{d}}x}  = 1} \\
& {\int_0^\infty  {x{f_X}\left( x \right){\rm{d}}x}  = \xi P}.
\end{split}
\label{equ37}
\end{eqnarray}
By solving problem (\ref{equ37}) using a variational method \cite{BIB25_1}, we obtain the following theorem.

\begin{theorem}
The optimal input PDF for the optimization problem (\ref{equ37}) is obtained as
\begin{eqnarray}
{f_X}\left( x \right) = \frac{1}{{\sqrt {1 + x{\varsigma ^2}} }}{e^{ - m - 1 - nx}},\quad x \ge 0, n>0,
\label{equ38}
\end{eqnarray}
where $m$ and $n$ are determined by the following two equalities
\begin{eqnarray}
\left\{ {\begin{array}{*{20}{c}}
{\frac{2}{{{\varsigma ^2}}}{e^{ - m - 1 + \frac{n}{{{\varsigma ^2}}}}}\sqrt {\frac{{\pi {\varsigma ^2}}}{n}} {\cal Q}\left( {\sqrt {\frac{{2n}}{{{\varsigma ^2}}}} } \right) = 1}\\
{\frac{1}{{{\varsigma ^2}n}}{e^{ - m - 1}} + \frac{1}{{2n}} - \frac{1}{{{\varsigma ^2}}} = \xi P}
\end{array}} \right..
\label{equ39}
\end{eqnarray}
\label{the5}
\end{theorem}

\begin{proof}
See Appendix \ref{Appendix_E}.
\end{proof}
Substituting (\ref{equ38}) into (\ref{equ13}), we obtain the following lower bound.

\begin{theorem}
For indoor VLC with constraints (\ref{equ3}) and (\ref{equ5}), a lower bound of the channel capacity is derived as\footnote{Similar to the lower bound in \emph{Theorem \ref{the2}}, the lower bound in \emph{Theorem \ref{the6}} is also valid for all values of optical intensities.}
\begin{equation}
C_{\rm low} = - \frac{1}{2}\ln \left( {2\pi e{\sigma ^2}} \right) + 1 + m + n\xi P + {f_{{\rm{low}}}}\left( {\xi P} \right),
\label{equ40}
\end{equation}
where $m$ and $n$ are the solution to (\ref{equ39}), and the function ${f_{{\rm{low}}}}\left( {\xi P} \right)$ is defined by (\ref{equ12}).
\label{the6}
\end{theorem}

\begin{proof}
See Appendix \ref{Appendix_F}.
\end{proof}

\subsection{Capacity Upper Bound}
In this subsection, the dual expression of capacity is also employed.
In this case, eqs. (\ref{equ26}) and (\ref{equ27}) can be derived as well.
To find a suitable upper bound, we choose $R_Y(y)$ as
\begin{eqnarray}
{R_Y}\left( y \right) = \left\{ {\begin{array}{*{20}{c}}
{\frac{{2\beta }}{{\sqrt {2\pi } }}{e^{ - \frac{{{y^2}}}{2}}}, \quad \quad \quad \quad \quad\quad\;\;\;  y \in \left( { - \infty ,0} \right)}\\
{\frac{{\left( {1 - \beta } \right){\varsigma ^2}{e^{ - ny}}}}{{2\sqrt {\frac{{\pi {\varsigma ^2}}}{n}} {e^{\frac{n}{{{\varsigma ^2}}}}}{\cal Q}\left( {\sqrt {\frac{{2n}}{{{\varsigma ^2}}}} } \right)\sqrt {1 + {\varsigma ^2}y} }},\; y \in \left[ {0,\infty } \right)}
\end{array}} \right..
\label{equ43}
\end{eqnarray}

According to (\ref{equ26}), (\ref{equ27}) and (\ref{equ43}), and the principle of ``capacity-achieving source distributions that escape to infinity", an upper bound in this case can be obtained in \emph{Theorem \ref{the7}}.

\begin{theorem}
For indoor VLC with constraints (\ref{equ3}) and (\ref{equ5}), an upper bound of the channel capacity is expressed as\footnote{Similar to the upper bound in \emph{Theorem \ref{the3}}, the upper bound in \emph{Theorem \ref{the7}} is also valid only when the optical intensity is sufficiently large.}
\begin{eqnarray}
C_{\rm upp}=  - \frac{1}{2}\ln \left( {2\pi e{\sigma ^2}} \right) + 1 + m + n\xi P + \ln \left( {\frac{1}{{1 - \beta }}} \right) + {o_P}\left( 1 \right),
\label{equ44}
\end{eqnarray}
where $m$ and $n$ are the solution to (\ref{equ39}), and $\beta \in (0,1) $ is arbitrary.
\label{the7}
\end{theorem}

\begin{proof}
See Appendix \ref{Appendix_G}.
\end{proof}

\subsection{Tightness of the Derived Bounds}
According to \emph{Theorem \ref{the6}} and \emph{Theorem \ref{the7}},
the performance gap between the upper bound and the lower bound is derived in \emph{Theorem \ref{the8}}.

\begin{theorem}
For indoor VLC with constraints (\ref{equ3}) and (\ref{equ5}), the performance gap between the upper bound in (\ref{equ44}) and the lower bound in (\ref{equ40}) is given by
\begin{eqnarray}
{C_{{\rm{gap}}}} &=& {C_{{\rm{upp}}}} - {C_{{\rm{low}}}} \nonumber \\
&=&   \ln \left( {\frac{1}{{1 - \beta }}} \right) + {o_P}\left( 1 \right)   - {f_{{\rm{low}}}}\left( {\xi P} \right).
\end{eqnarray}
\label{the8}
\end{theorem}

\begin{proof}
The proof is straightforward and is omitted.
\end{proof}

According to \emph{Theorem {\ref{the8}}}, for large $P$, the following corollary is derived.
\begin{corollary}
For VLC with constraints (\ref{equ3}) and (\ref{equ5}), the asymptotic performance gap between upper bound (\ref{equ44}) and lower bound (\ref{equ40}) is given by
\begin{equation}
\mathop {\lim }\limits_{P \to \infty } {C_{\rm gap}} = \ln \left( {\frac{1}{{1 - \beta }}} \right).
\label{equ45}
\end{equation}
\label{cor2}
\end{corollary}

\begin{proof}
The proof of this corollary is straightforward and is omitted.
\end{proof}

\begin{remark}
Because $\beta$ is a small positive number, the asymptotic performance gap
in \emph{Corollary \ref{cor2}} is small enough to be ignored.
This indicates that upper bound (\ref{equ44}) and lower bound (\ref{equ40}) are asymptotically tight for large $P$.
\label{rem2}
\end{remark}

\section{Numerical Results}
\label{section5}
In this section, some representative results will be provided.
The accuracy of the derived capacity bounds will also be verified.
Without loss of generality,
the variance of the signal-independent noise is normalized to be one (i.e., $\sigma^2=1$),
both $\beta$ and $\delta$ are set to be 0.001.
Because practical VLC systems operate at high SNRs in indoor illumination environments, $A\; ({\rm or}\; P) \ge 30\; {\rm{dB}}$ is defined as the target VLC operation range.

\subsection{Results of VLC With Both Average and Peak Optical Intensity Constraints }
In this subsection, the lower bound (\ref{a1}) in \emph{Theorem \ref{the2}} and the upper bound (\ref{a2}) in \emph{Theorem \ref{the3}} will be verified.

Fig. \ref{fig2} plots the derived channel capacity bounds versus $A$ with different peak-to-nominal-optical-intensity ratios $A/P$ when $\xi  = 0.3$ and ${\varsigma ^2} = 1.5$.
Obviously, all capacity bounds monotonously increase with $A$.
For a small fixed $A$ value (i.e., $A\leq 50 \;{\rm dB}$), by increasing the value of $A/P$, $P$ decreases accordingly,
and thus all capacity bounds decrease.
However, when $A> 50 \;{\rm dB}$, the value of $A/P$ has almost no impact on capacity performance.
This is because all the terms related to $P$ in (\ref{a1}) and (\ref{a2}) disappear when $A$ is large. As it is known, the value of  $A/P$ reflects the luminous ability of an LED. The above phenomenon implies that the type of LED does not affect system performance when the optical intensity is large. That is, any type of high-optical-intensity LED can be adopted for indoor VLC.
Moreover, the upper bound and the lower bound are loose for small values of $A$.
The reason is that the upper bound is derived based on the concept of ``capacity-achieving source distributions that escape to infinity", and thus the upper bound is asymptotically valid only. In other words, the derived upper bound of capacity can be inaccurate for small values of $A$.
Fortunately, we only care about the capacity of VLC at high SNR.
Under such an operation condition, the upper bound and the lower bound asymptotically coincide,
which verifies the conclusion derived in \emph{Remark {\ref{rem1}}}.

\begin{figure}
\centering
\includegraphics[width=9.3cm]{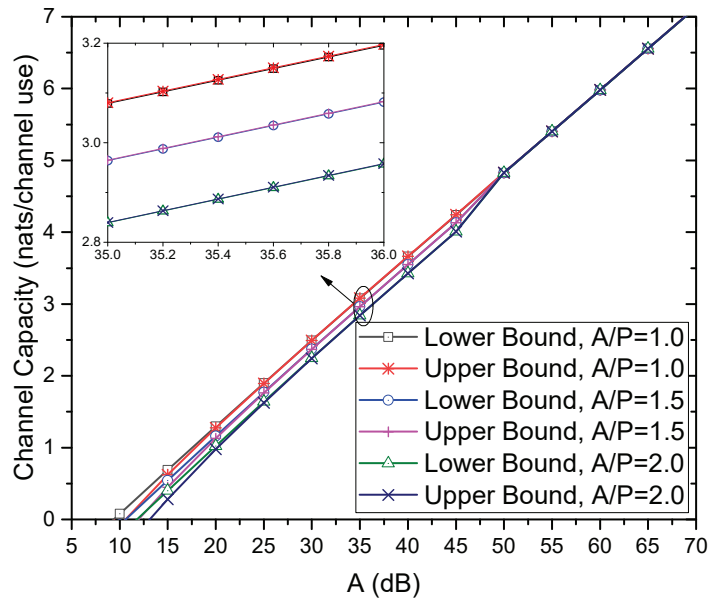}
\caption{Capacity bounds versus $A$ with different $A/P$ when $\xi  = 0.3$ and ${\varsigma ^2} = 1.5$ \label{fig2}}
\end{figure}

Fig. \ref{fig3} plots the channel capacity bounds versus $\xi $ with different values of $A/P$ when ${\varsigma ^2} = 1.5$ and $A = 45\;{\rm{dB}}$.
As seen, as the increase of $\xi $,
all channel capacity bounds first increase rapidly and then tend to stabilized values.
This indicates that the dimming target $\xi$ has great influences on the channel capacity.
Moreover, the channel capacity bounds decrease with $A/P$.
Once again, for all the capacity curves, the gaps between the upper bound and the lower bound are small, which can be ignored.

\begin{figure}
\centering
\includegraphics[width=9.3cm]{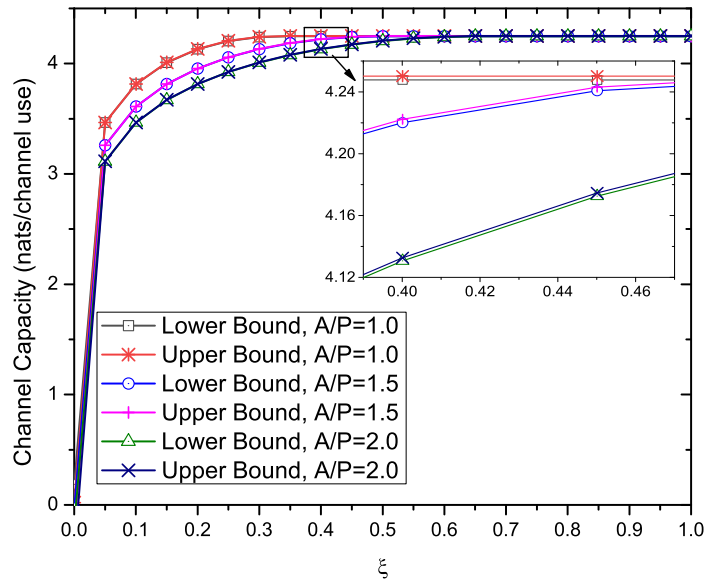}
\caption{Capacity bounds versus $\xi$ with different $A/P$ when ${\varsigma ^2} = 1.5$ and $A = 45\;{\rm{dB}}$ \label{fig3}}
\end{figure}

\begin{figure}
\centering
\includegraphics[width=9.3cm]{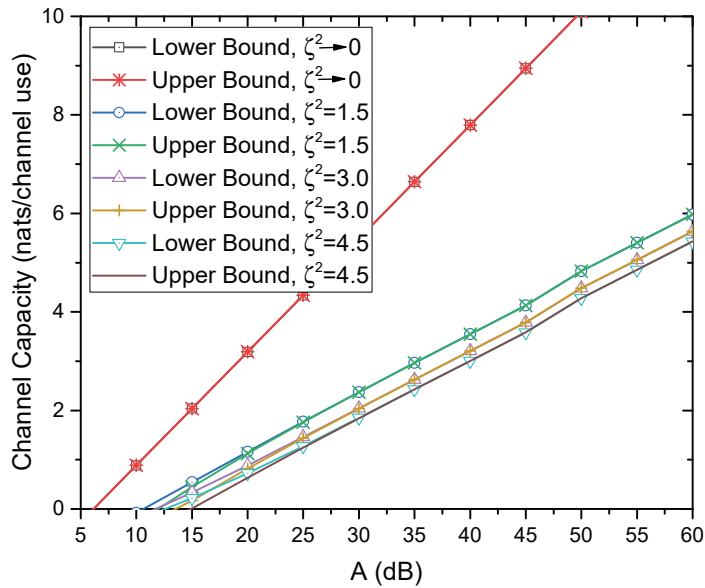}
\caption{Capacity bounds versus $A$ with different ${\varsigma ^2}$ when $\xi  = 0.3$ and $A = 1.5P$ \label{fig4}}
\end{figure}

Fig. \ref{fig4} plots the channel capacity bounds versus $A$ with different values of ${\varsigma ^2}$ when $\xi  = 0.3$ and $A = 1.5P$.
For comparison, the results of the capacity bounds when ${\varsigma ^2} \to 0$,
which corresponds to the scenario with only the signal-independent noise, are also presented.
It can be observed that the signal-dependent Gaussian noise greatly affects the system performance.
Specifically, the curves with ${\varsigma ^2} \to 0$ achieve the largest channel capacity bounds.
Moreover,  all channel capacity bounds decrease with ${\varsigma ^2}$.
Once again, the capacity lower and upper bounds are tight under target operation conditions.
Furthermore, as $A$ tends to infinity, the gaps between the upper bound and the lower bound asymptotically approach zero.

\subsection{Results of VLC Only With Average Optical Intensity Constraint}
In this subsection, the lower bound (\ref{equ40}) in \emph{Theorem \ref{the6}} and the upper bound (\ref{equ44}) in \emph{Theorem \ref{the7}} will be verified.

Fig. \ref{fig5} plots the channel capacity bounds versus $P$ with different values of $\xi$ when ${\varsigma ^2} =1.5$.
Obviously, with the increase of $P$, the capacity bounds for all curves increase monotonously.
For a fixed $P$, all channel capacity bounds increase with $\xi$.
For a larger dimming target $\xi$, a larger average optical intensity can be produced, and the ability of the channel to transmit information can be improved.
Moreover, for a fixed $\xi$, the lower bound and the upper bound almost coincide with each other when $P$ is large.
This verifies the conclusion in \emph{Remark \ref{rem2}}.

\begin{figure}
\centering
\includegraphics[width=9.3cm]{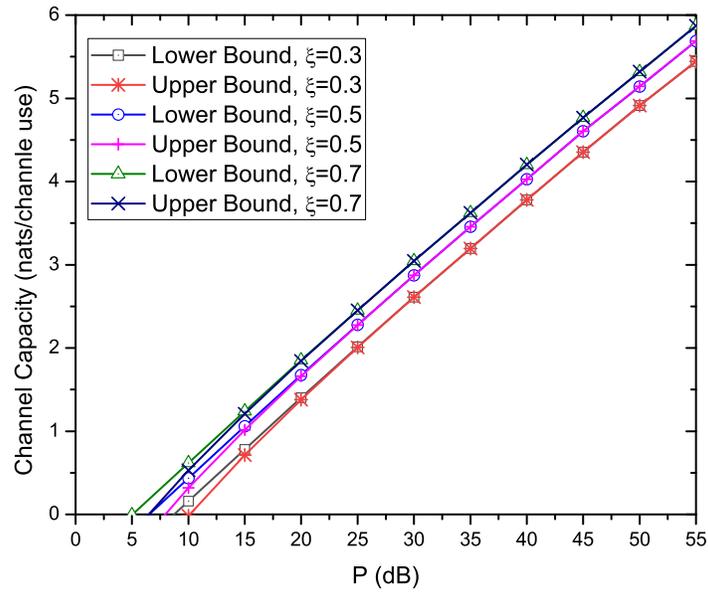}
\caption{Capacity bounds versus $P$ with different $\xi$ when ${\varsigma ^2} =1.5$ \label{fig5}}
\end{figure}

\begin{figure}
\centering
\includegraphics[width=9.3cm]{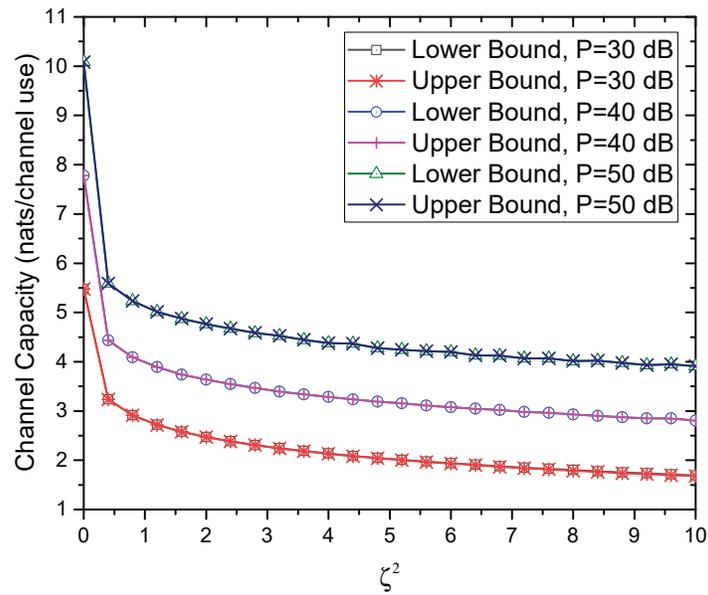}
\caption{Capacity bounds versus $\varsigma ^2$ with different $P$ when $\xi = 0.3$ \label{fig6}}
\end{figure}

Fig. \ref{fig6} plots the channel capacity bounds versus $\varsigma ^2$ with different values of $P$ when $\xi =0.3$.
As we can see, the channel capacity bounds decrease monotonically with $\varsigma ^2$, which is consistent with the conclusion in Fig. \ref{fig4}.
Moreover, the capacity performance for VLC improves with $P$.
To quantify, when the nominal optical intensity is increased by 10 dB,
channel capacity is increased by 1.1 nats.
Once again, the lower bound and the upper bound on channel capacity are tight,
which verifies \emph{Remark \ref{rem2}}.

\subsection{Comparisons With Existing Results}
In this subsection, for the indoor VLC with signal-dependent AWGN, we compare
the capacity bounds having both average and peak optical intensity constraints with the capacity bounds having only the average optical intensity constraint.
As known, with only signal-independent noise, the capacity bounds of VLC having both average and peak optical intensity constraints were presented in \cite{BIB17}, while the capacity bounds of VLC having only the average optical intensity constraint were derived in \cite{BIB18}. Therefore, previously reported capacity bounds in \cite{BIB17} and \cite{BIB18} are presented here. Moreover, the classic Shannon's capacity over AWGN channel \cite{BIB25} is also provided for comparison.

\begin{figure}
\centering
\includegraphics[width=10cm]{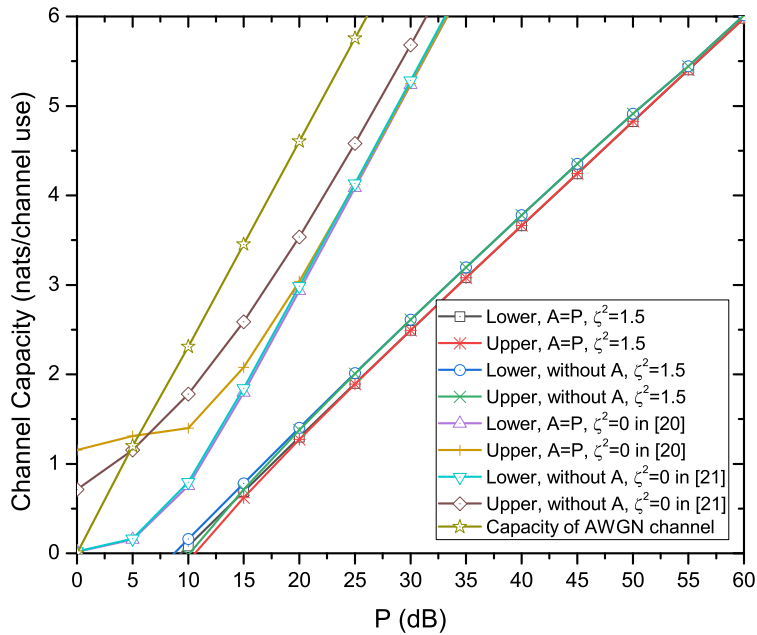}
\caption{Capacity bounds versus $P$ for different scenarios when $\xi  = 0.3$ \label{fig7}}
\end{figure}

Fig. \ref{fig7} plots the channel capacity bounds versus $P$ for different scenarios when $\xi  = 0.3$.
Similar to the previous figures, all capacity bounds increase monotonically with $P$.
Compared with the signal-independent noise in \cite{BIB17} and \cite{BIB18},
the signal-dependent noise in this paper degrades the capacity performance.
This indicates that the signal-dependent noise has a large impact on the capacity performance.
Moreover, the VLC without peak optical intensity constraint can obtain a slightly larger capacity than that with peak optical intensity constraint.
This shows that the peak optical intensity of the LED results in a channel capacity loss.
Furthermore, the capacity of AWGN channel \cite{BIB25} is larger than that of the VLC channel.

\begin{figure}
\centering
\includegraphics[width=10cm]{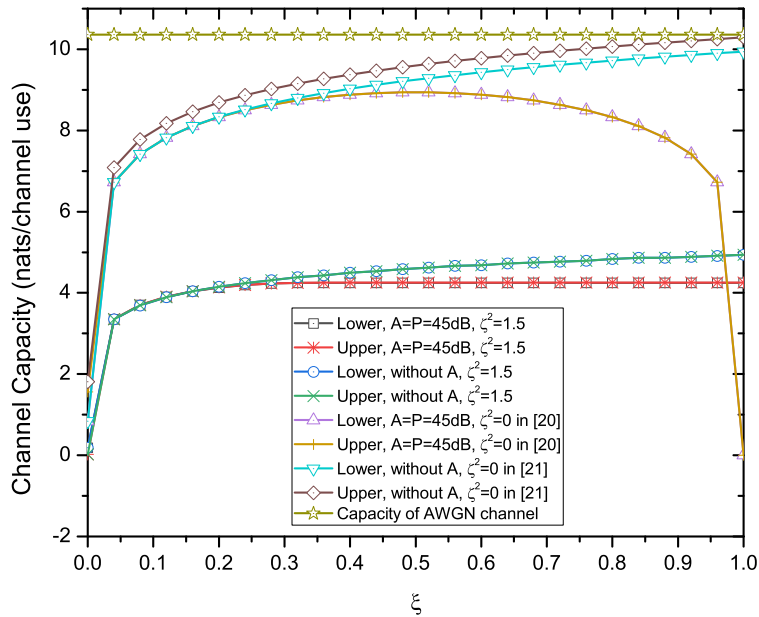}
\caption{Capacity bounds versus $\xi $ for different scenarios \label{fig8}}
\end{figure}

Fig. \ref{fig8} shows the channel capacity bounds versus $\xi $ for different VLC scenarios.
The capacity over the AWGN channel \cite{BIB25} is a constant because the dimming support is not considered in conventional RF wireless communications. Nevertheless, the AWGN channel always achieves the largest capacity.
For the VLC with only signal-independent noise (i.e., $\varsigma ^2=0$) and without peak optical intensity constraint in \cite{BIB17},
the second best capacity bounds are obtained.
Moreover, in this case, the capacity bounds increase with $\xi$.
For the indoor VLC with $\varsigma ^2=0$ and $A=P$ in \cite{BIB18},
the capacity bounds increase and then decrease with $\xi $,
and the capacity bound curves are symmetrical with $\xi  = 0.5$.
For the VLC with signal-dependent noise,
the capacity bounds increase with $\xi$.
As seen, the curves with $A=P=45 \;{\rm dB}$ and $\varsigma ^2=1.5$ have the worst capacity performance.

\section{Conclusions}
\label{section6}
This paper investigated the channel capacity bounds for indoor VLC with signal-dependent Gaussian noise.
The chief conclusions of this paper are summarized below:
\begin{itemize}
\item For the considered two scenarios in Sections \ref{section3} and \ref{section4},
capacity lower and upper bounds for VLC are derived, which can be employed to evaluate the performance of VLC efficiently.
\item  In practical VLC operation regime, the derived capacity lower and upper bounds are tight,
and thus we can accurately estimate the channel capacity of VLC.
\item When the optical intensity is large, the peak-to-nominal-optical-intensity ratio has no effect on the capacity performance.
\item The signal-dependent noise in VLC has a strong influence on capacity performance.
When the variance of signal-dependent noise is increased, the capacity performance degrades.
\item Compared to that with peak optical intensity constraint,
the system without considering the peak optical intensity constraint can achieve improved performance.
\end{itemize}

Since the signal-dependent noise plays an important role in VLC system at high SNR, we will explore an experimental VLC platform to validate the proposed signal-dependent noise model in our future research.

\numberwithin{equation}{section}
\appendices
\section{Justification of the Model in (\ref{equ1})}
\label{Appendix_A}
Fig. \ref{fig9} shows the equivalent circuit for the VLC receiver in Fig. \ref{fig1}.
In this figure, there are three kinds of shot noise current sources, which originate from carriers generated by signal photons $\sqrt {\overline {i_{\rm{s}}^2} }$, background photons $\sqrt {\overline {i_{\rm{b}}^2} }$, and dark current $\sqrt {\overline {i_{\rm{d}}^2} }$.
These shot noise current sources can be expressed as \cite{BIB23}
\begin{eqnarray}
\left\{ \begin{array}{l}
 \sqrt{\overline {i_{\rm{s}}^2}}  = \sqrt{2\frac{{{q^2}\eta }}{{h\nu }}XB}\\
 \sqrt{\overline {i_{\rm{b}}^2}}  = \sqrt{2\frac{{{q^2}\eta }}{{h\nu }}{X_{\rm{b}}}B}\\
 \sqrt{\overline {i_{\rm{d}}^2}}  = \sqrt{2q{i_{\rm{d}}}B}\\
 \end{array} \right.,
\label{equ34}
\end{eqnarray}
where $q$ is the electron charge,
$\eta $ is the quantum efficiency,
$\nu$ is the signal frequency,
$h$ is Planck constant,
$B$ is the bandwidth,
and ${X_{\rm{b}}}$ is the background optical intensity.

The bias and load resistances have a thermal noise current source $\sqrt {\overline {i_{{\rm{th}}}^2} }$ as \cite{BIB23}
\begin{eqnarray}
\sqrt{\overline {i_{{\rm{th}}}^2}}  = \sqrt{\frac{{4KTB}}{{{R_{\rm{e}}}}}},
\label{equ35}
\end{eqnarray}
where $K$ is the Boltzmann constant, $T$ is the absolute temperature, ${R_{\rm{e}}}$ is the parallel combination of ${R_{\rm{B}}}$ and ${R_{\rm{L}}}$.

\begin{figure}
\centering
\includegraphics[width=12cm]{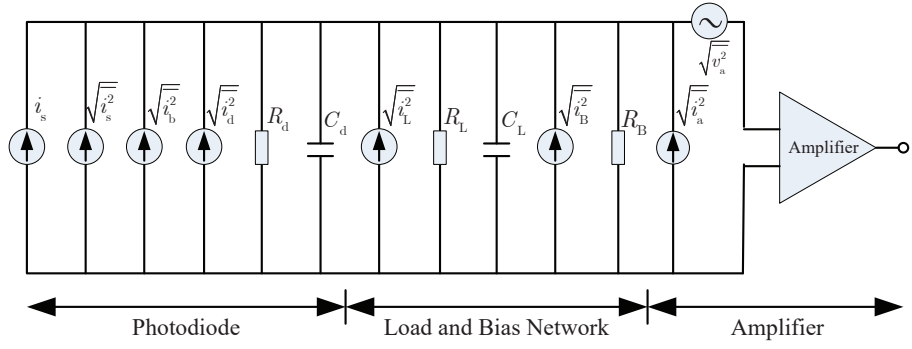}
\caption{The equivalent circuit for the VLC receiver in Fig. \ref{fig1} \label{fig9}}
\end{figure}

Assume that a high impedance amplifier is employed at the receiver, whose first stage is made of a bipolar junction transistor (BJT). The amplifier noise current source $\sqrt {\overline {i_{\rm{a}}^2} }$ is given by \cite{BIB23}
\begin{eqnarray}
\sqrt{\overline {i_{\rm{a}}^2}}  = \sqrt{\frac{{2KTB}}{{{r_{\rm{e}}}}}},
\label{equ36}
\end{eqnarray}
where ${r_{\rm{e}}}$ is the alternating current input resistance of the BJT.

Summarizing (\ref{equ34}), (\ref{equ35}) and (\ref{equ36}) collectively, we can find that $\sqrt {\overline {i_{\rm{s}}^2} }$ depends on $\sqrt X $,
while $\sqrt {\overline {i_{\rm{b}}^2} }$, $\sqrt {\overline {i_{\rm{d}}^2} }$, $\sqrt {\overline {i_{{\rm{th}}}^2} }$ and $\sqrt {\overline {i_{\rm{a}}^2} }$ are all independent of $\sqrt X$. That is, the additive noises can be classified into two types: input-dependent noise and input-independent noise.
Moreover, according to \cite{BIB08}, \cite{BIB10} and \cite{BIB22}, the shot noise, thermal noise, and amplifier noise can be well modelled by Gaussian distributions.
As a result, the channel model in (\ref{equ1}) is justified for practical VLC systems.

\section{Proof of \emph{Theorem \ref{the1}}}
\label{Appendix_B}
Let $f_X(x)$ be the solution to problem (\ref{equ14}),
and then let ${\tilde f}_X(x) = f_X(x) + \varepsilon \eta(x)$ be a perturbation function of $f_X(x)$,
where $\varepsilon $ is a small number and $\eta(x)$ is a differentiable function of $x$.
Note that ${\tilde f}_X(x)$ should also meet the two constraints of problem (\ref{equ14}).
Equivalently, $\eta(x)$ should satisfy
\begin{eqnarray}
\left\{ \begin{array}{l}
\int_0^A {\eta (x){\rm{d}}x}  = 0\\
\int_0^A {x\eta (x){\rm{d}}x}  = 0
\end{array} \right..
\label{equ:A2}
\end{eqnarray}

Define $\rho (\varepsilon ) = {\cal J}\left[ {f_X(x) + \varepsilon \eta (x)} \right]$ as a function of $\varepsilon$.
Because the minimum value of $\rho (\varepsilon )$ is attained when $\varepsilon  = 0$, and thus we obtain
\begin{equation}
{\left. {\frac{{{\rm{d}}\rho (\varepsilon )}}{{{\rm{d}}\varepsilon }}} \right|_{\varepsilon  = 0}} = \int_0^A {\eta (x)\left\{ \ln ( {f_X(x)} ) + 1+ \frac{1}{2}\ln (1 + {\varsigma ^2}x)
 \right\}{\rm{d}}x}  = 0.
\label{equ:A4}
\end{equation}

Combining (\ref{equ:A4}) with (\ref{equ:A2}), we obtain the following equation
\begin{equation}
\ln ( {{f_X}(x)} ) + 1 + \frac{1}{2}\ln (1 + {\varsigma ^2}x) = c+bx,
\label{equ:A6}
\end{equation}
where $c$ and $b$ are arbitrary constants.

\textbf{Case 1:} If $b=0$, the PDF of $X$ is $f_X(x) = \frac{1}{{\sqrt {1 + {\varsigma ^2}x} }}{{e}^{c - 1}},\;0 \le x \le A$.
Substituting the PDF into the first constraint of (\ref{equ14}), we have
\begin{equation}
{e^{c-1}} = \frac{{{\varsigma ^2}}}{{2\left( {\sqrt {1 + {\varsigma ^2}A}  - 1} \right)}}.
\label{equ:A9}
\end{equation}
Therefore, the first equation in (\ref{b1}) holds.
Substituting the first equation in (\ref{b1}) into the second constraint of (\ref{equ14}), one has
\begin{eqnarray}
\xi P = \frac{{{\varsigma ^2}A + \sqrt {1 + {\varsigma ^2}A}  - 1}}{{3{\varsigma ^2}}}.
\label{equ:A11}
\end{eqnarray}

According to the definition of APR, eq. (\ref{equ:A11}) can be equivalently written as
\begin{eqnarray}
\alpha  = \frac{{{\varsigma ^2}A + \sqrt {1 + {\varsigma ^2}A}  - 1}}{{3{\varsigma ^2}A}}.
\label{equ:A12}
\end{eqnarray}

\textbf{Case 2:} If $b\neq0$, the PDF of $X$ is $f_X(x) = \frac{1}{{\sqrt {1 + {\varsigma ^2}x} }}{{e}^{bx+c - 1}},\;0 \le x \le A$.
Substituting the PDF into the first constraint of (\ref{equ14}), we have
\begin{eqnarray}
{e^{c-1}} = \frac{{{\varsigma ^2}}}{{2g\left( b,{\varsigma ^2},A \right)}}.
\label{equ:A13}
\end{eqnarray}
Substituting (\ref{equ:A13}) into the PDF of $X$, the second equation in (\ref{b1}) holds.
Then, substituting the second equation in (\ref{b1}) into the second constraint of (\ref{equ14}), eq. (\ref{equ20}) holds.

\section{Proof of \emph{Theorem \ref{the2}}}
\label{Appendix_C}
\textbf{Case 1:} $\alpha  = \left( {{\varsigma ^2}A + \sqrt {1 + {\varsigma ^2}A}  - 1} \right) / \left( {3{\varsigma ^2}A} \right)$. In this case, $H(X)$ is given by
\begin{eqnarray}
H(X) = \ln \left( {\frac{2}{{{\varsigma ^2}}}\left( {\sqrt {1 + {\varsigma ^2}A}  - 1} \right)} \right) + \frac{1}{2}E_X\left( {{\rm{ln}}\left( {1 + {\varsigma ^2}X} \right)} \right).
\label{equ:C1}
\end{eqnarray}
Substituting (\ref{equ:C1}) into (\ref{equ13}), the lower bound of capacity is obtained as the first equation in (\ref{a1}).

\textbf{Case 2:} $\alpha  \ne \left( {{\varsigma ^2}A + \sqrt {1 + {\varsigma ^2}A}  - 1} \right) / \left( {3{\varsigma ^2}A} \right)$ and $\alpha  \in (0, P/A]$.
In this case, we have
\begin{eqnarray}
H\left( X \right) &=& \!\!\int_0^A {{f_X}\left( x \right)\ln \left( {\frac{{2g\left( {b,{\varsigma ^2},A} \right)}}{{{\varsigma ^2}}}} \right){\rm{d}}x}\!  + \! \frac{1}{2}\int_0^A {{f_X}\left( x \right)\ln \left( {1 + {\varsigma ^2}x} \right){\rm{d}}x} \! -\! b\int_0^A {{f_X}\left( x \right)x\,{\rm{d}}x} \nonumber \\
 &=& \ln \left( {\frac{{2g\left( {b,{\varsigma ^2},A} \right)}}{{{\varsigma ^2}}}} \right) - b\xi P + \frac{1}{2}{E_X}\left( {\ln \left( {1 + {\varsigma ^2}X} \right)} \right).
\label{equ:C3}
\end{eqnarray}
Substituting (\ref{equ:C3}) into (\ref{equ13}), the capacity can be lower-bounded by the second equation in (\ref{a1}).

\section{Proof of \emph{Theorem \ref{the3}}}
\label{Appendix_D}
\textbf{Case 1:} $\alpha  = \left( {{\varsigma ^2}A + \sqrt {1 + {\varsigma ^2}A}  - 1} \right) / \left( {3{\varsigma ^2}A} \right)$.
In this case, substitute (\ref{equ28}) into (\ref{equ27}), we obtain
\begin{eqnarray}
&& \!\!\!\! D\left( {{f_{Y|X}}\left( {y\left| x \right.} \right)\left\| {{R_Y}\left( y \right)} \right.} \right) = - \frac{1}{2}\ln \left( {2\pi e{\sigma ^2}\left( {1 + {\varsigma ^2}X} \right)} \right)\underbrace { - \int_{ - \infty }^0 {{f_{Y|X}}\left( {y\left| x \right.} \right)\ln \left(\frac{{2\beta }}{{\sqrt {2\pi } }}{e^{ - \frac{{{y^2}}}{2}}}\right)} {\rm{d}}y}_{ \buildrel \Delta \over = {D_1}} \nonumber \\
&& \!\!\!\!\underbrace { -\!\!\! \int_0^{A \!+\! A\delta }\!\! \!\!\!\!{{f_{Y|X}}\!( {y\left| x \right.} )\!\ln\!\! \left(\!\!\frac{{\left( {1 - 2\beta } \right){\varsigma ^2}}}{{2\!\!\left(\!\! {\sqrt {1 \!+\! {\varsigma ^2}A( {1 \!+\! \delta } )} \! -\! 1}\! \right)\!\!\sqrt {1 \!+\! {\varsigma ^2}y} }}\!\!\right) } \!{\rm{d}}y}_{ \buildrel \Delta \over = {D_2}} \underbrace { - \int_{A \!+\! A\delta }^{ \infty }\!\! {{f_{Y|X}}\!( {y\left| x \right.})\!\ln\!\! \left(\!\!\frac{\beta {e^{ - y}}}{{{e^{ - \left( {A \!+\! A\delta } \right)}}}}\!\!\right)} \!{\rm{d}}y}_{ \buildrel \Delta \over = {D_3}},
\label{equ:D1}
\end{eqnarray}
where $D_1$ can be expressed as
\begin{eqnarray}
{D_1} \!\!\!\!&=& \!\!\!\!\left[ \!{\!\frac{{\left( {1 \!+\! X{\varsigma ^2}} \right){\sigma ^2} \!+\! {X^2}}}{2} + \ln \frac{{\sqrt {2\pi } }}{{2\beta }}} \!\right]\!\!{\cal Q}\!\!\left(\! {\frac{X}{{\sqrt {\left( {1 \!+\! X{\varsigma ^2}} \right){\sigma ^2}} }}} \!\right)\! - \!\frac{{ X \!\sqrt {\!\left( {1 \!+\! X{\varsigma ^2}} \right){\sigma ^2}} }}{{2\sqrt {2\pi } }}{e^{^{ - \frac{{{X^2}}}{{2\left( {1 \!+\! X{\varsigma ^2}} \right){\sigma ^2}}}}}} \nonumber \\
 &=& \!\!\!\!{o_X}\left( 1 \right),
\label{equ:D2}
\end{eqnarray}

Moreover, $D_2$ in (\ref{equ:D1}) can be expressed as
\begin{eqnarray}
{D_2} &=& \underbrace {\ln \left( {\frac{{2\left( {\sqrt {1 + {\varsigma ^2}A\left( {1 + \delta } \right)}  - 1} \right)}}{{\left( {1 - 2\beta } \right){\varsigma ^2}}}} \right)}_{ \ge 0,\;{\rm{for \;large \;A}}}\underbrace {\int_0^{A\left( {1 + \delta } \right)} {\frac{{{e^{ - \frac{{{{\left( {y - X} \right)}^2}}}{{2\left( {1 + X{\varsigma ^2}} \right){\sigma ^2}}}}}}}{{\sqrt {2\pi \left( {1 + X{\varsigma ^2}} \right){\sigma ^2}} }}{\rm{d}}y} }_{ \le 1} \nonumber \\
 && + \frac{1}{2}\underbrace {\int_0^{A\left( {1 + \delta } \right)} {\frac{{\ln \left( {1 + {\varsigma ^2}y} \right)}}{{\sqrt {2\pi \left( {1 + X{\varsigma ^2}} \right){\sigma ^2}} }}{e^{^{ - \frac{{{{\left( {y - X} \right)}^2}}}{{2\left( {1 + X{\varsigma ^2}} \right){\sigma ^2}}}}}}{\rm{d}}y} }_{ \buildrel \Delta \over = {D_4}} \nonumber \\
 &\le& \ln \left( {\frac{{2\left( {\sqrt {1 + {\varsigma ^2}A\left( {1 + \delta } \right)}  - 1} \right)}}{{\left( {1 - 2\beta } \right){\varsigma ^2}}}} \right) + \frac{1}{2}{D_4},
\label{equ:D3}
\end{eqnarray}
where $\ln \left( 2[ {\sqrt {1 \!+\! \varsigma ^2 A ( 1 + \delta )}  - 1} ] / [{\left( {1 - 2\beta } \right){\varsigma ^2}}] \right) \ge 0$ holds because $A$ is large in VLC \cite{BIB20,BIB21}.

For the term $D_4$ in (\ref{equ:D3}), if $X \ge {1 \mathord{\left/ {\vphantom {1 {{\varsigma ^2}}}} \right. \kern-\nulldelimiterspace} {{\varsigma ^2}}}$, we have
\begin{eqnarray}
{D_4}  &=&\!\! \underbrace {\ln \left( {{\varsigma ^2}X} \right)}_{ \ge 0}\!\!\int_{\frac{{ - X}}{{\sqrt {\left( {1 + X{\varsigma ^2}} \right){\sigma ^2}} }}}^{\frac{{A + A\delta  - X}}{{\sqrt {\left( {1 + X{\varsigma ^2}} \right){\sigma ^2}} }}} {\frac{{{e^{ - \frac{{{h^2}}}{2}}}}}{{\sqrt {2\pi } }}{\rm{d}}h} \!\! + \!\!\frac{X}{{\sqrt {2\pi \left( {1 + X{\varsigma ^2}} \right){\sigma ^2}} }}\!\!\int_{\frac{{\frac{1}{{{\varsigma ^2}}} - X}}{X}}^{\frac{{\frac{1}{{{\varsigma ^2}}} + A + A\delta  - X}}{X}}\!\! {\underbrace {\ln \left( {1 + \tilde y} \right)}_{ \le \tilde y}{e^{ - \frac{{{{\left( {X\tilde y - \frac{1}{{{\varsigma ^2}}}} \right)}^2}}}{{2\left( {1 + X{\varsigma ^2}} \right){\sigma ^2}}}}}{\rm{d}}\tilde y} \nonumber\\
 &\le& \ln \left( {{\varsigma ^2}X} \right) + \frac{X}{{\sqrt {2\pi \left( {1 + X{\varsigma ^2}} \right){\sigma ^2}} }}\int_{\frac{{\frac{1}{{{\varsigma ^2}}} - X}}{X}}^{\frac{{\frac{1}{{{\varsigma ^2}}} + A + A\delta  - X}}{X}} {\tilde y{e^{ - \frac{{{{\left( {X\tilde y - \frac{1}{{{\varsigma ^2}}}} \right)}^2}}}{{2\left( {1 + X{\varsigma ^2}} \right){\sigma ^2}}}}}{\rm{d}}\tilde y} \nonumber\\
 &=& \ln \left( {{\varsigma ^2}X} \right) + \frac{{\sqrt {\left( {1 + X{\varsigma ^2}} \right){\sigma ^2}} }}{{\sqrt {2\pi } X}}\left[ {{e^{ - \frac{{{X^2}}}{{2\left( {1 + X{\varsigma ^2}} \right){\sigma ^2}}}}} - \underbrace {{e^{ - \frac{{{{\left( {A + A\delta  - X} \right)}^2}}}{{2\left( {1 + X{\varsigma ^2}} \right){\sigma ^2}}}}}}_{ \ge 0}} \right] \nonumber \\
&&\; + \frac{1}{{X{\varsigma ^2}}}\underbrace {\left[ {{\cal Q}\left( {\frac{{ - X}}{{\sqrt {\left( {1 + X{\varsigma ^2}} \right){\sigma ^2}} }}} \right) - {\cal Q}\left( {\frac{{A + A\delta  - X}}{{\sqrt {\left( {1 + X{\varsigma ^2}} \right){\sigma ^2}} }}} \right)} \right]}_{ \le 1} \nonumber \\
 &\le& \ln \left( {{\varsigma ^2}X} \right) + \frac{{\sqrt {\left( {1 + X{\varsigma ^2}} \right){\sigma ^2}} }}{{\sqrt {2\pi } X}}{e^{ - \frac{{{X^2}}}{{2\left( {1 + X{\varsigma ^2}} \right){\sigma ^2}}}}} + \frac{1}{{X{\varsigma ^2}}}.
\label{equ:D4}
\end{eqnarray}

If $0 \le X < {1 \mathord{\left/ {\vphantom {1 {{\varsigma ^2}}}} \right. \kern-\nulldelimiterspace} {{\varsigma ^2}}}$, then we can bound $\ln \left( {1 + {\varsigma ^2}y} \right) \le {\varsigma ^2}y$ and obtain
\begin{eqnarray}
{D_4} \!\!\!\!&\le&\!\!\!\! \int_0^{A + A\delta } {\frac{{{\varsigma ^2}y}}{{\sqrt {2\pi \left( {1 + X{\varsigma ^2}} \right){\sigma ^2}} }}{e^{ - \frac{{{{\left( {y - X} \right)}^2}}}{{2\left( {1 + X{\varsigma ^2}} \right){\sigma ^2}}}}}{\rm{d}}y} \nonumber\\
 &=&\!\!\!\! {\varsigma ^2}\!\sqrt {\!\frac{{( {1 \!\!+\!\! X{\varsigma ^2}} ){\sigma ^2}}}{{2\pi }}} \!\!\left[\! {{e^{ - \frac{{{X^2}}}{{2\left( {1 \!+\! X{\varsigma ^2}} \right){\sigma ^2}}}}} \!-\! \underbrace {{e^{ - \frac{{{{\left( {A \!+\! A\delta  \!-\! X} \right)}^2}}}{{2\left( {1 \!+\! X{\varsigma ^2}} \right){\sigma ^2}}}}}}_{ \ge 0}} \!\right]
 \!\!+\!\! {\varsigma ^2}\!X\!\! \!\left[\! {{\cal Q}\!\!\left(\!\! {\frac{{ - X}}{{\sqrt {\!( {1 \!\!+\!\! X{\varsigma ^2}})\!{\sigma ^2}} }}}\!\! \right) \!\!-\!\! \underbrace {{\cal Q}\!\!\left(\!\! {\frac{{A \!+\! A\delta  \!-\! X}}{{\sqrt {\!( {1 \!+\! X{\varsigma ^2}} )\!{\sigma ^2}} }}}\!\! \right)}_{ \ge 0}} \!\right] \nonumber \\
 &\le&\!\!\!\! {\varsigma ^2}\sqrt {\frac{{\left( {1 + X{\varsigma ^2}} \right){\sigma ^2}}}{{2\pi }}} {e^{ - \frac{{{X^2}}}{{2\left( {1 + X{\varsigma ^2}} \right){\sigma ^2}}}}} + {\varsigma ^2}X {\cal Q}\left( {\frac{{ - X}}{{\sqrt {\left( {1 + X{\varsigma ^2}} \right){\sigma ^2}} }}} \right).
  \label{equ:D5}
\end{eqnarray}

Because (\ref{equ:D5}) is bounded and from (\ref{equ:D4}), we obtain
\begin{eqnarray}
{D_4} \le \ln \left( {{\varsigma ^2}X} \right) + \frac{{\sqrt {\left( {1 + X{\varsigma ^2}} \right){\sigma ^2}} }}{{\sqrt {2\pi } X}}{e^{ - \frac{{{X^2}}}{{2\left( {1 + X{\varsigma ^2}} \right){\sigma ^2}}}}} + \frac{1}{{X{\varsigma ^2}}}.
\label{equ:D6}
\end{eqnarray}

Substituting (\ref{equ:D6}) into (\ref{equ:D3}), we have
\begin{equation}
{D_2} \!\!\le\!\! \ln\! \left( {\frac{{2\left( {\sqrt {1 \!+\! {\varsigma ^2}A\left( {1 \!+\! \delta } \right)}  \!-\! 1} \right)}}{{\left( {1 - 2\beta } \right){\varsigma ^2}}}} \!\right) \!+\! \frac{1}{2}\left[ {\ln \left( {{\varsigma ^2}X} \right) \!+\! \frac{{\sqrt {\left( {1 \!+\! X{\varsigma ^2}} \right){\sigma ^2}} }}{{\sqrt {2\pi } X}}{e^{ - \frac{{{X^2}}}{{2\left( {1 \!+\! X{\varsigma ^2}} \right){\sigma ^2}}}}} \!+\! \frac{1}{{X{\varsigma ^2}}}} \right].
\label{equ:D7}
\end{equation}

Furthermore, $D_3$ in (\ref{equ:D1}) can be expressed as
\begin{eqnarray}
{D_3} &=& \left[ {X - \left( {A + A\delta  + \ln \beta } \right)} \right]{\cal Q}\left( {\frac{{A + A\delta  - X}}{{\sqrt {\left( {1 + X{\varsigma ^2}} \right){\sigma ^2}} }}} \right) + \frac{{\sqrt {\left( {1 + X{\varsigma ^2}} \right){\sigma ^2}} }}{{\sqrt {2\pi } }}{e^{ - \frac{{{{\left( {A + A\delta  - X} \right)}^2}}}{{2\left( {1 + X{\varsigma ^2}} \right){\sigma ^2}}}}} \nonumber \\
 &\le&  - \left( {A\delta  + \ln \beta } \right){\cal Q}\left( {\frac{{A + A\delta }}{{\sqrt {\left( {1 + A{\varsigma ^2}} \right){\sigma ^2}} }}} \right) + \frac{{\sqrt {\left( {1 + A{\varsigma ^2}} \right){\sigma ^2}} }}{{\sqrt {2\pi } }}{e^{ - \frac{{{{\left( {A + A\delta } \right)}^2}}}{{2\left( {1 + A{\varsigma ^2}} \right){\sigma ^2}}}}} \nonumber \\
 &=& {o_A}\left( 1 \right).
  \label{equ:D10}
\end{eqnarray}

Substituting (\ref{equ:D2}), (\ref{equ:D7}) and (\ref{equ:D10}) into (\ref{equ:D1}), $D\left( {{f_{Y|X}}\left( {y\left| x \right.} \right)\left\| {{R_Y}\left( y \right)} \right.} \right)$ can be derived.
Then, substituting $D\left( {{f_{Y|X}}\left( {y\left| x \right.} \right)\left\| {{R_Y}\left( y \right)} \right.} \right)$ into (\ref{equ26}), we have
\begin{eqnarray}
\!\!\!C \!\!\!\! &\le&\!\!\!\! \ln \!\!\left( {\frac{{2\left( {\sqrt {1 + {\varsigma ^2}A\left( {1 + \delta } \right)}  - 1} \right)}}{{\left( {1 - 2\beta } \right){\varsigma ^2}}}} \right) \!\!-\! \frac{1}{2}\ln \left( {2\pi e{\sigma ^2}} \right) \!+\! {E_{X:{f_X}\left( x \right) = f_X^*\left( x \right)}}\!\left[ {{o_X}\!\left( 1 \right)} \right]\! +\! {o_A}\!\left( 1 \right).
\label{equ:D12}
\end{eqnarray}

To solve (\ref{equ:D12}) further, the concept of ``capacity-achieving source distributions that escape to infinity" is utilized \cite{BIB22,BIB28,BIB29}.
According to this concept, under the circumstances of the capacity-achieving input PDF ${f_X^*}\left( x \right)$,
when the optical intensity approaches infinity, the probability for any set of finite-intensity input symbols will tend to zero.
Consequently, we obtain
\begin{eqnarray}
{E_{{X:{f_X}\left( x \right) = f_X^*\left( x \right)}}} ( {{o_X}\left( 1 \right)} ) = {o_A}\left( 1 \right).
\label{equ:D13}
\end{eqnarray}
Substitute (\ref{equ:D13}) into (\ref{equ:D12}), we obtain the first equation in (\ref{a2}).

\textbf{Case 2:} $\alpha  \ne {\left( {{\varsigma ^2}A + \sqrt {1 + {\varsigma ^2}A}  - 1} \right)} / \left( {3{\varsigma ^2}A} \right)$ and $\alpha  \in (0,P/A]$. In this case, substitute (\ref{equ29}) into (\ref{equ27}), we have
\begin{eqnarray}
\!\!\!\!\!\!\!\!\!\!\! D\left( {{f_{Y|X}}\left( {y\left| x \right.} \right)\left\| {{R_Y}\left( y \right)} \right.} \right)\!\!&=&\!\!  - \frac{1}{2}\ln \left( {2\pi e{\sigma ^2}\left( {1 + {\varsigma ^2}X} \right)} \right)+D_1+D_3 \nonumber \\
&&\underbrace { - \int_0^{A + A\delta } {{f_{Y|X}}\left( {y\left| x \right.} \right)\ln \left( \frac{{\left( {1 - 2\beta } \right){\varsigma ^2}{e^{by}}}}{{2G\left( {b,{\varsigma ^2},A,\delta } \right)\sqrt {1 + {\varsigma ^2}y} }}\right)} \,{\rm{d}}y}_{ \buildrel \Delta \over = {D_5}} .
\label{equ:D15}
\end{eqnarray}

For the term $D_5$ in (\ref{equ:D15}), we have
\begin{eqnarray}
{D_5} \!\!\!\!&=&\!\!\!\! \ln\! \left(\! {\frac{{2G\!( b,{\varsigma ^2},A,\delta  )}}{{( {1 \!-\! 2\beta } ){\varsigma ^2}}}}\! \right)\!\! \underbrace {\int_0^{A( 1 \!+\! \delta )}\!\! {\frac{{{e^{ - \frac{{{{\left( {y - X} \right)}^2}}}{{2\left( {1 + X{\varsigma ^2}} \right){\sigma ^2}}}}}}}{{\sqrt {2\pi \!(1 \!\!+\!\! X{\varsigma ^2}){\sigma ^2}} }}{\rm{d}}y} }_{ \le 1} + \frac{1}{2}D_4\! +\! \underbrace {\int_0^{A( 1 \!+\! \delta )}\!\!\! {\frac{{ - by}}{{\sqrt {2\pi ( {1 \!\!+\!\! X{\varsigma ^2}}){\sigma ^2}} }}{e^{^{ - \frac{{{{\left( {y - X} \right)}^2}}}{{2( {1 \!+\! X{\varsigma ^2}}){\sigma ^2}}}}}}\!{\rm{d}}y} }_{ \buildrel \Delta \over = {D_6}} \nonumber \\
 &\le& \ln \left( {\frac{{2G\left( {b,{\varsigma ^2},A,\delta } \right)}}{{\left( {1 - 2\beta } \right){\varsigma ^2}}}} \right) + \frac{1}{2}{D_4} + {D_6},
\label{equ:D16}
\end{eqnarray}
where $\ln \left( {2G( {b,{\varsigma ^2},A,\delta } )} /[( {1 - 2\beta }){\varsigma ^2}] \right) \ge 0$ holds because $A$ is large in VLC \cite{BIB20,BIB21}.

For the term ${D_6}$, when $b < 0$, we have
\begin{eqnarray}
{D_6} \!\!\!\! &=&\!\!\!\!  - b\!\! \left[\! \underbrace X_{ \ge 0}\!\!\underbrace {\left[\! {\cal Q}\!\!\left(\!\! {\frac{{ - X}}{{\sqrt {( 1 \!\!+\!\! X{\varsigma ^2} ){\sigma ^2}} }}} \!\!\right) \!\!-\!\! {\cal Q}\!\!\left(\!\! {\frac{{A \!+\! A\delta  \!-\! X}}{{\sqrt {( 1 \!\!+\!\! X{\varsigma ^2}){\sigma ^2}} }}} \!\!\right)\!\! \right]}_{ \le 1} \!+\! \frac{{\sqrt {( 1 \!\!+\!\! X{\varsigma ^2}){\sigma ^2}} }}{{\sqrt {2\pi } }}\!\!\!\left(\!\! {{e^{ - \frac{{{X^2}}}{{2( {1 \!+\! X{\varsigma ^2}}){\sigma ^2}}}}} \!-\! \underbrace {{e^{\! - \frac{{{{( A \!+\! A\delta  \!-\! X )}^2}}}{{2( {1 \!+\! X{\varsigma ^2}}){\sigma ^2}}}}}}_{ \ge 0}}\!\! \right)\!\! \right] \nonumber \\
 &\le&\!\!\!\!  - bX - b\frac{{\sqrt {\left( {1 + A{\varsigma ^2}} \right){\sigma ^2}} }}{{\sqrt {2\pi } }}{e^{ - \frac{{{A^2}}}{{2\left( {1 + A{\varsigma ^2}} \right){\sigma ^2}}}}}.
\label{equ:D17}
\end{eqnarray}

When $b > 0$, we have
\begin{eqnarray}
{D_6} \!\!\!\!&=& \!\!\!\! \left| { - b} \right|\left[ \frac{{\sqrt {\left( {1 + X{\varsigma ^2}} \right){\sigma ^2}} }}{{\sqrt {2\pi } }}\left( {{e^{ - \frac{{{{\left( {A + A\delta  - X} \right)}^2}}}{{2\left( {1 + X{\varsigma ^2}} \right){\sigma ^2}}}}} - \underbrace {{e^{ - \frac{{{X^2}}}{{2\left( {1 + X{\varsigma ^2}} \right){\sigma ^2}}}}}}_{ \ge 0}} \right) \right. \nonumber \\
 && \!\!\!\! \left.- \underbrace X_{ \ge 0}\left( {\underbrace {{\cal Q}\left( {\frac{{ - X}}{{\sqrt {\left( {1 + X{\varsigma ^2}} \right){\sigma ^2}} }}} \right) -{\cal Q}\left( {\frac{{A + A\delta  - X}}{{\sqrt {\left( {1 + X{\varsigma ^2}} \right){\sigma ^2}} }}} \right)}_{ > 0}} \right) \right] \nonumber \\
 &\le&\!\!\!\! b\frac{{\sqrt {\left( {1 \!+\! A{\varsigma ^2}} \right){\sigma ^2}} }}{{\sqrt {2\pi } }}{e^{ - \frac{{{{\left( {A\delta } \right)}^2}}}{{2\left( {1 + A{\varsigma ^2}} \right){\sigma ^2}}}}}\!\! -\! bX\!\!\left[ {{\cal Q}\!\left( \!{\frac{{ - X}}{{\sqrt {\left( {1 \!+\! X{\varsigma ^2}} \right){\sigma ^2}} }}} \!\right)\!\! -\! {\cal Q}\!\left(\! {\frac{{A + A\delta  - X}}{{\sqrt {\left( {1 \!+\! X{\varsigma ^2}} \right){\sigma ^2}} }}} \!\right)} \!\right].
\label{equ:D18}
\end{eqnarray}

Define
\begin{eqnarray}
\varphi \!\left( {b,\!{\varsigma ^2},\!A,\!{\sigma ^2},\!X} \right)\!\! \buildrel \Delta \over = \!\!\left\{ {\begin{array}{*{20}{c}}\!\!\!
{ - bX - b\frac{{\sqrt {\left( {1 + A{\varsigma ^2}} \right){\sigma ^2}} }}{{\sqrt {2\pi } }}{e^{ - \frac{{{A^2}}}{{2\left( {1 + A{\varsigma ^2}} \right){\sigma ^2}}}}} ,\; \; \; \; \; b \!<\! 0}\\
{\!\!\!\!b\frac{{\sqrt {\left( {1 \!+\! A{\varsigma ^2}} \right){\sigma ^2}} }}{{\sqrt {2\pi } }}{e^{ - \frac{{{{\left( {A\delta } \right)}^2}}}{{2\left( {1 \!+\! A{\varsigma ^2}} \right){\sigma ^2}}}}} \!\!-\! bX\!\!\left[ {{\cal Q}\!\!\left( \!{\frac{{ - X}}{{\sqrt {\left( {1 \!+\! X{\varsigma ^2}} \right){\sigma ^2}} }}} \!\right)\!\! -\! {\cal Q}\!\!\left(\! {\frac{{A + A\delta  - X}}{{\sqrt {\left( {1 \!+\! X{\varsigma ^2}} \right){\sigma ^2}} }}} \!\right)} \!\right]\!\!,\! b \!>\! 0}
\end{array}} \right..
\label{equ:D19}
\end{eqnarray}
Combine (\ref{equ:D17}) with (\ref{equ:D18}), $D_6$ can be further expressed as
\begin{eqnarray}
{D_6} \le \varphi \left( {b,{\varsigma ^2},A,{\sigma ^2},X} \right).
\label{equ:D20}
\end{eqnarray}

Then, substitute (\ref{equ:D6}) and (\ref{equ:D20}) into (\ref{equ:D16}), $D_5$ can be written as
\begin{eqnarray}
{D_5} \!\!\le\!\! \ln\!\! \left(\! {\frac{{2G( {b,{\varsigma ^2},A,\delta } )}}{{\left( {1 \!-\! 2\beta } \right){\varsigma ^2}}}} \!\right) \!+\! \frac{1}{2}\!\!\left[\! {\ln ( {{\varsigma ^2}X} ) \!+\! \frac{{\sqrt {\!( {1 \!\!+\!\! X{\varsigma ^2}} ){\sigma ^2}} }}{{\sqrt {2\pi } X}}{e^{ - \frac{{{X^2}}}{{2( {1 \!+\! X{\varsigma ^2}}){\sigma ^2}}}}} \!\!+\!\! \frac{1}{{X{\varsigma ^2}}}}\!\! \right]  \!\!+\!\! \varphi \!( {b,{\varsigma ^2},A,{\sigma ^2},X} ).
\label{equ:D21}
\end{eqnarray}

Substituting (\ref{equ:D21}) into (\ref{equ:D15}), we can derive $D(\! {{f_{Y|X}}( {y\left| x \right.} )\!\left\| {{R_Y}\!( y)} \right.} \!)$.
Then, substituting the derived relative entropy $D(\! {{f_{Y|X}}( {y\left| x \right.} )\!\left\| {{R_Y}\!( y)} \right.} \!)$ into (\ref{equ:D18}), we have
\begin{eqnarray}
C  \!\!\!\!&\le& \!\!\!\! \ln\!\! \left(\! {\frac{{2G\!( {b,\!{\varsigma ^2}\!,A,\delta\! })}}{{\left( {1 \!-\! 2\beta } \right){\varsigma ^2}}}}\! \right) \!\!-\!\! \frac{1}{2}\!\ln\! ( \!{2\pi e{\sigma ^2}}\! ) \!+\! \psi\!( {b,\!{\varsigma ^2}\!,A,{\sigma ^2}\!,\xi ,\!P} ) \!\!+\! E_{X\!:{f_X}\!( x) = f_X^*\left( x \right)}[o_X\!(1)]\!\!+\! {o_A}\!( 1),
\label{equ:D23}
\end{eqnarray}
where $\psi \left( {b,{\varsigma ^2},A,{\sigma ^2},\xi ,P} \right)$ is given by
\begin{eqnarray}
&&\!\!\!\!\!\!\psi \left( {b,{\varsigma ^2},A,{\sigma ^2},\xi ,P} \right)\buildrel \Delta \over = {E_{X:{f_X}\left( x \right) = f_X^*\left( x \right)}}\left[ {\varphi \left( {b,{\varsigma ^2},A,{\sigma ^2},X} \right)} \right] \nonumber \\
 &&=\!\! \left\{ {\begin{array}{*{20}{c}}
\!\!\!\!{ - b\frac{{\sqrt {\left( {1 + A{\varsigma ^2}} \right){\sigma ^2}} }}{{\sqrt {2\pi } }}{e^{ - \frac{{{A^2}}}{{2\left( {1 + A{\varsigma ^2}} \right){\sigma ^2}}}}} - b\xi P,\;\;\;\;b < 0}\\
\!\!\!\!{b\frac{{\sqrt {\left( {1 + A{\varsigma ^2}} \right){\sigma ^2}} }}{{\sqrt {2\pi } }}{e^{ - \frac{{{{\left( {A\delta } \right)}^2}}}{{2\left( {1 + A{\varsigma ^2}} \right){\sigma ^2}}}}}\!\!-\!b\xi P\!\!\left[ {{\cal Q}\!\!\left(\!\! {\frac{{ - \xi P}}{{\sqrt {\left( {1 + \xi P{\varsigma ^2}} \right){\sigma ^2}} }}} \!\!\right) \!\!-\! {\cal Q}\!\!\left(\!\! {\frac{{A + A\delta  - \xi P}}{{\sqrt {\left( {1 + \xi P{\varsigma ^2}} \right){\sigma ^2}} }}}\!\! \right)} \right], b > 0}
\end{array}} \right..
\label{equ:D24}
\end{eqnarray}
After that, substitute (\ref{equ:D13}) into (\ref{equ:D23}), the second equation in (\ref{a2}) holds.

\section{Proof of \emph{Corollary \ref{coro1} }}
\label{appd}
\textbf{Case 1:} $\alpha  = \left( {{\varsigma ^2}A + \sqrt {1 + {\varsigma ^2}A}  - 1} \right) / \left( {3{\varsigma ^2}A} \right)$. By using \emph{Theorem \ref{the4}} and the L'hospital's rule, we have
\begin{eqnarray}
\mathop {\lim }\limits_{A \to \infty } {C_{{\rm{gap}}}}  &=& \ln \left( {\frac{{\sqrt {1 + \delta } }}{{1 - 2\beta }}} \right) - \mathop {\lim }\limits_{A \to \infty } {f_{{\rm{low}}}}\left( {\xi P} \right).
 \label{appd1}
\end{eqnarray}
Owing to $\alpha A = \xi P$, we have
$
\mathop {\lim }\limits_{A \to \infty } {f_{{\rm{low}}}}(\xi P)  = 0
$. Then, the first equation in \emph{Corollary \ref{cor1}} holds.

\textbf{Case 2:} $\alpha  \ne \left( {{\varsigma ^2}A + \sqrt {1 + {\varsigma ^2}A}  - 1} \right) / \left( {3{\varsigma ^2}A} \right)$ and $\alpha \in (0,P/A]$. In this case, we have
\begin{eqnarray}
\mathop {\lim }\limits_{A \to \infty } {C_{{\rm{gap}}}} \!\!\!\! &=& \ln \left( {\frac{1}{{1 - 2\beta }}} \right) + \mathop {\lim }\limits_{A \to \infty } \left[ {\psi \left( {b,{\varsigma ^2},A,\xi ,P} \right) + b\xi P} \right] - \mathop {\lim }\limits_{A \to \infty } {f_{{\rm{low}}}}\left( {\xi P} \right).
\label{appd3}
\end{eqnarray}
Moreover, when $b<0$, we have
\begin{eqnarray}
\mathop {\lim }\limits_{A \to \infty } \left[ {\psi \left( {b,{\varsigma ^2},A,\xi ,P} \right) + b\xi P} \right] = \mathop {\lim }\limits_{A \to \infty } \left[ { - b\frac{{\sqrt {\left( {1 + A{\varsigma ^2}} \right){\sigma ^2}} }}{{\sqrt {2\pi } }}{e^{ - \frac{{{A^2}}}{{2\left( {1 + A{\varsigma ^2}} \right){\sigma ^2}}}}}} \right] = 0.
\label{appd4}
\end{eqnarray}
When $b \geq0$, we have
\begin{eqnarray}
\lim_{A \to \infty } [ {\psi \left( {b,{\varsigma ^2},A,\xi ,P} \right) + b\xi P} ] \!\!\!\!&=&\!\!\!\! \lim_{A \to \infty } \left\{ b\frac{{\sqrt {\left( {1 + A{\varsigma ^2}} \right){\sigma ^2}} }}{{\sqrt {2\pi } }}{e^{ - \frac{{{{\left( {A\delta } \right)}^2}}}{{2\left( {1 + A{\varsigma ^2}} \right){\sigma ^2}}}}} \right. \nonumber\\
&+&\!\!\!\! \left. b\xi P\!\left[ {1 \!+\! {\cal Q}\!\left(\! {\frac{{A \!+\! A\delta  \!-\! \xi P}}{{\sqrt {\left( {1 \!+\! \xi P{\varsigma ^2}} \right){\sigma ^2}} }}} \!\right) \!-\! {\cal Q}\!\left(\! {\frac{{ - \xi P}}{{\sqrt {\left( {1 \!+\! \xi P{\varsigma ^2}} \right){\sigma ^2}} }}} \!\right)}\! \right]\! \right\} \nonumber \\
 &=& 0.
\label{appd5}
\end{eqnarray}
Substituting (\ref{appd4}) and (\ref{appd5}) into (\ref{appd3}), we derive the second equation in \emph{Corollary \ref{cor1}}.

\section{Proof of \emph{Theorem \ref{the5}}}
\label{Appendix_E}
Let ${f_X}\left( x \right)$ be the solution to problem (\ref{equ37}) and ${\tilde f_X}\left( x \right) = {f_X}\left( x \right) + \varepsilon \eta \left( x \right)$ be a perturbation function of ${f_X}\left( x \right)$, which satisfies the two constraints in (\ref{equ37}), so $\eta \left( x \right)$ should satisfy
\begin{eqnarray}
\left\{ \begin{array}{l}
\int_0^\infty  {\eta \left( x \right){\rm{d}}x}  = 0\\
\int_0^\infty  {x\eta \left( x \right){\rm{d}}x}  = 0
\end{array} \right..
\label{equ:E4}
\end{eqnarray}

Define a function $\rho \left( \varepsilon  \right) = {\cal J}\left[ {{f_X}\left( x \right) + \varepsilon \eta \left( x \right)} \right]$. Referring to Appendix \ref{Appendix_B}, we have
\begin{eqnarray}
\left. {\frac{{{\rm d}\rho \left( \varepsilon  \right)}}{{{\rm d}\varepsilon }}}\right| _{\varepsilon  = 0}  = \int_0^\infty  {\eta \left( x \right)\left\{ {\ln ( {{f_X}\left( x \right)} ) + 1 + \frac{1}{2}\ln \left( {1 + {\varsigma ^2}x} \right)} \right\}{\rm{d}}x}  = 0.
\label{equ:E6}
\end{eqnarray}

Taking the two constraints in (\ref{equ:E4}) into account, we obtain the following equation
\begin{eqnarray}
\ln ( {{f_X}\left( x \right)} ) + 1 + \frac{1}{2}\ln \left( {1 + {\varsigma ^2}x} \right) =  -m - nx,
\label{equ:E7}
\end{eqnarray}
where $m$ and $n$ are arbitrary constants.

\textbf{Case 1:} If $n=0$, the input PDF is expressed as ${f_X}( x ) = \frac{1}{{\sqrt {1 + x{\varsigma ^2}} }}{e^{ -m - 1}},\; x>0$.
Substituting the input PDF into the first constraint of problem (\ref{equ37}), we have
\begin{eqnarray}
{e^{-m - 1}}\frac{2}{{{\varsigma ^2}}}\left(\left. {\sqrt {1 + {\varsigma ^2}x} }\; \right|_0^\infty \right) = 1.
\end{eqnarray}
This indicates that the above equality has no solution.

\textbf{Case 2:} If $n<0$, the input PDF is expressed as ${f_X}\left( x \right) = \frac{1}{{\sqrt {1 + x{\varsigma ^2}} }}{e^{- m - 1 - nx}},\; x>0, n<0$.
Substituting the input PDF into the first constraint of problem (\ref{equ37}), we have
\begin{eqnarray}
\int_0^\infty  {{f_X}(x){\rm{d}}x}
 = \frac{{2{e^{-m - 1}}}}{{{\varsigma ^2}}}\int_0^\infty  {{e^{-n\frac{{{t^2} - 1}}{{{\varsigma ^2}}}}}{\rm{d}}t}  = 1.
\end{eqnarray}
This indicates that the above equality has no solution.

\textbf{Case 3:} If $n>0$, the input PDF is expressed as ${f_X}\left( x \right) = \frac{1}{{\sqrt {1 + x{\varsigma ^2}} }}{e^{- m - 1 - nx}},\; x>0, n>0$.
Substituting the input PDF into the first constraint of problem (\ref{equ37}), we have
\begin{eqnarray}
\int_0^\infty  {{f_X}\left( x \right){\rm{d}}x}  =\frac{2}{{{\varsigma ^2}}}{e^{ -m - 1 + \frac{n}{{{\varsigma ^2}}}}}\sqrt {\frac{{\pi {\varsigma ^2}}}{n}} {\cal Q}\left( {\sqrt {\frac{{2n}}{{{\varsigma ^2}}}} } \right) = 1.
\label{equ:E9}
\end{eqnarray}

Moreover, substituting the input PDF into the second constraint of (\ref{equ37}), we can obtain
\begin{eqnarray}
\int_0^\infty \!\! {x{f_X}\!( x){\rm{d}}x}  \!\!=\!\! \frac{{e^{ - m - 1 + \frac{n}{{{\varsigma ^2}}}}}}{{{\varsigma ^2}n}}\!\left[\! {{e^{\frac{{ - n}}{{{\varsigma ^2}}}}} \!+\! \sqrt {\frac{{\pi {\varsigma ^2}}}{n}} {\cal Q}\left( {\sqrt {\frac{{2n}}{{{\varsigma ^2}}}} } \right)} \!\right]\! -\! \frac{2{e^{ - m - 1 + \frac{n}{{{\varsigma ^2}}}}}}{{{\varsigma ^4}}}\sqrt {\frac{{\pi {\varsigma ^2}}}{n}} {\cal Q}\!\left( \!{\sqrt {\frac{{2n}}{{{\varsigma ^2}}}} } \right).
\label{equ:E10}
\end{eqnarray}
According to (\ref{equ:E9}), eq. (\ref{equ:E10}) is further simplified as
\begin{eqnarray}
\int_0^\infty  {x{f_X}\left( x \right){\rm{d}}x}  = \frac{1}{{{\varsigma ^2}n}}{e^{ - m - 1}} + \frac{1}{{2n}} - \frac{1}{{{\varsigma ^2}}} = \xi P.
\end{eqnarray}
Hence, combining all three cases, \emph{Theorem \ref{the5}} holds.

\section{Proof of \emph{Theorem \ref{the6}}}
\label{Appendix_F}
Substituting the input PDF (\ref{equ38}) into (\ref{equ13}), we write the lower bound on channel capacity as
\begin{eqnarray}
C &\ge&  - \int_0^\infty  {{f_X}\left( x \right)\ln \left( {\frac{1}{{\sqrt {1 + x{\varsigma ^2}} }}{e^{ - m - 1 - nx}}} \right){\rm{d}}x} \nonumber \\
 && - \frac{1}{2}\int_0^\infty  {{f_X}\left( x \right)\ln \left( {1 + x{\varsigma ^2}} \right){\rm{d}}x}  + {f_{{\rm{low}}}}\left( {\xi P} \right) - \frac{1}{2}\ln \left( {2\pi e{\sigma ^2}} \right) \nonumber \\
 &=& \left( {1 + m} \right)\int_0^\infty  {{f_X}\left( x \right){\rm{d}}x}  + n\int_0^\infty  {x{f_X}\left( x \right){\rm{d}}x}  + {f_{{\rm{low}}}}\left( {\xi P} \right) - \frac{1}{2}\ln \left( {2\pi e{\sigma ^2}} \right).
\label{equ:F1}
\end{eqnarray}
Substituting the constraints of (\ref{equ37}) into (\ref{equ:F1}), \emph{Theorem \ref{the6}} holds.

\section{Proof of \emph{Theorem \ref{the7}}}
\label{Appendix_G}
Substituting (\ref{equ43}) into (\ref{equ27}), we have
\begin{eqnarray}
D\left( {{f_{Y|X}}\left( {y\left| X \right.} \right)\left\| {{R_Y}\left( y \right)} \right.} \right) \!\!&=& \!\! - \frac{1}{2}\ln \!\left( {2\pi e{\sigma ^2}\left( {1 + {\varsigma ^2}X} \right)} \right)+D_1 \nonumber \\
&&\underbrace { \!\!-\!\! \int_0^\infty \!\! {{f_{Y|X}}\!\left( {y\left| X \right.} \!\right)\ln \frac{{\left( {1 - \beta } \right){\varsigma ^2}{e^{ - ny}}}}{{2\sqrt {\frac{{\pi {\varsigma ^2}}}{n}} {e^{\frac{n}{{{\varsigma ^2}}}}}{\cal Q}\left( {\sqrt {\frac{{2n}}{{{\varsigma ^2}}}} } \right)\sqrt {1 + {\varsigma ^2}y} }}} \,\!{\rm{d}}y}_{ \buildrel \Delta \over = {D_7}}.
\label{equ:G1}
\end{eqnarray}

From (\ref{equ:D2}), we know that ${D_1} = {o_X}\left( 1 \right)$.
For ${D_7}$, we have
\begin{eqnarray}
{D_7} \!\!\!\! &=&\!\!\!\! \ln\!\! \left(\! {\frac{{2\sqrt {\!\frac{{\pi {\varsigma ^2}}}{n}} {e^{\frac{n}{{{\varsigma ^2}}}}}{\cal Q}\!\!\left(\! {\sqrt {\!\frac{{2n}}{{{\varsigma ^2}}}} } \right)}}{{\left( {1 \!-\! \beta } \right){\varsigma ^2}}}} \!\!\right) \!\!\underbrace {\int_0^\infty \!\!\! {{f_{Y|X}}( {y\left| X \right.}\! )} {\rm{d}}y}_{ \le 1} \! + \frac{1}{2}\!\underbrace {\int_0^\infty \!\!\! {\ln ( {1 \!\!+\!\! {\varsigma ^2}y} ){f_{Y|X}}( {y\left| X \right.} \!)} {\rm{d}}y}_{{\triangleq D_8}}+ n\!\!\underbrace {\int_0^\infty  \!\!\!\!{y{f_{Y|X}}\!( {y\left| X \right.}\! )} {\rm{d}}y}_{{\triangleq D_9}} \nonumber \\
 &\le& \ln \left( {\frac{{2\sqrt {\frac{{\pi {\varsigma ^2}}}{n}} {e^{\frac{n}{{{\varsigma ^2}}}}}{\cal Q}\left( {\sqrt {\frac{{2n}}{{{\varsigma ^2}}}} } \right)}}{{\left( {1 - \beta } \right){\varsigma ^2}}}} \right) + \frac{1}{2}\left[ {\ln \left( {{\varsigma ^2}X} \right) + \frac{{\sqrt {\left( {1 + X{\varsigma ^2}} \right){\sigma ^2}} }}{{\sqrt {2\pi } X}}{e^{ - \frac{{{X^2}}}{{2\left( {1 + X{\varsigma ^2}} \right){\sigma ^2}}}}} + \frac{1}{{X{\varsigma ^2}}}} \right] \nonumber \\
 &&+ n\left[ {X + \frac{{\sqrt {\left( {1 + X{\varsigma ^2}} \right){\sigma ^2}} }}{{\sqrt {2\pi } }}{e^{ - \frac{{{X^2}}}{{2\left( {1 + X{\varsigma ^2}} \right){\sigma ^2}}}}}} \right],
\label{equ:G2}
\end{eqnarray}
where ${D_8}$ and ${D_9}$ can be obtained by referring to the derivations of ${D_4}$ and ${D_6}$.

Substituting (\ref{equ:D2}) and (\ref{equ:G2}) into (\ref{equ:G1}), the relative entropy is derived.
Then, substitute (\ref{equ:G1}) into (\ref{equ26}), we get
\begin{eqnarray}
C \le- \frac{1}{2}\ln \left( {2\pi e{\sigma ^2}} \right) + \ln \left( {\frac{{2\sqrt {\frac{{\pi {\varsigma ^2}}}{n}} {e^{\frac{n}{{{\varsigma ^2}}}}}{\cal Q}\left( {\sqrt {\frac{{2n}}{{{\varsigma ^2}}}} } \right)}}{{\left( {1 - \beta } \right){\varsigma ^2}}}} \right) + {E_{X:{f_X}\left( x \right) = f_X^*\left( x \right)}}\left( {{o_X}\left( 1 \right) + nX} \right).
\label{equ:G3}
\end{eqnarray}
Moreover, we can get ${E_{X:{f_X}\left( x \right) = f_X^*\left( x \right)}} (o_X( 1 ))=o_P(1) $ and ${E_{X:{f_X}\left( x \right) = f_X^*\left( x \right)}}( nX )=n \xi P$. From (\ref{equ39}), \emph{Theorem \ref{the7}} can be finally derived.


\begin{thebibliography}{1}
\baselineskip=5.6mm
\bibitem{BIB21_1}
J.-Y. Wang, J.-B. Wang, M. Chen, and J. Wang, ``Capacity bounds for dimmable visible light communications using PIN photodiodes with input-dependent Gaussian noise," in \emph{IEEE Global Commun. Conf.}, Austin, TX, USA, 2014, pp. 2066-2071.

\bibitem{BIB01}
A. G. Bell, ``The photophone," \emph{Science}, vol. 1, no. 11, pp. 130-134, Sep. 1880.

\bibitem{BIB02}
J.-Y. Wang, C. Liu, J.-B. Wang, Y. Wu, M. Lin, and J. Cheng, ``Physical-layer security for indoor visible light communications: Secrecy capacity analysis," \emph{IEEE Trans. Commun.}, vol. 66, no. 12, pp. 6423-6436, Dec. 2018.

\bibitem{BIB03}
S. Ma, H. Li, Y. He, R. Yang, S. Lu, W. Cao, and S. Li, ``Capacity bounds and interference management for interference channel in visible light communication networks," \emph{IEEE Trans. Wirel. Commun.}, vol. 18, no. 1, pp. 182-193, Jan. 2019.

\bibitem{BIB04}
L. Jia, F. Shu, N. Huang, M. Chen, and J. Wang, ``Capacity and optimum signal constellations for VLC systems," \emph{J. Lightwave Technol.}, vol. 38, no. 8, pp. 2180-2189, Apr. 2020.

\bibitem{BIB15}
B. Li, W. Xu, S. Feng, and Z. Li, ``Spectral-efficient reconstructed LACO-OFDM transmission for dimming compatiable visible light communications," \emph{IEEE Photon. J.}, vol. 11, no. 1, pp. 1-14, Feb. 2019.

\bibitem{BIB16}
K.-I. Ahn and J. K. Kwon, ``Capacity analysis of M-PAM inverse source coding in visible light communication," \emph{J. Lightwave Technol.}, vol. 30, no. 10, pp. 1399-1404, May 2012.

\bibitem{BIB20}
J. Grubor, S. Randel, K.-D. Langer, and J. Walewski, ``Broadband information broadcasting using LED-based interior lighting," \emph{J. Lightwave Technol.}, vol. 26, no. 24, pp. 3883-3892, Dec. 2008.

\bibitem{BIB21}
L. Hanzo, H. Haas, S. Imre, D. O'Brien, M. Rupp, and L. Gyongyosi, ``Wireless myths, realities, and futures: From 3G/4G to optical and quantum wireless," \emph{Proc. IEEE}, vol. 100, no. Special Centennial Issue, pp. 1853-1888, May 2012.


\bibitem{BIB05}
S. M. Hass and J. H. Shaprio, ``Capacity of wireless optical communications," \emph{IEEE J. Sel. Areas Commun.}, vol. 21, no. 8, pp. 1346-1357, Oct. 2003.

\bibitem{BIB06}
K. Chakraborty and P. Narayan, ``The Poisson fading channel," \emph{IEEE Trans. Inf. Theory}, vol. 53, no. 7, pp. 2349-2364, July 2007.

\bibitem{BIB07}
K. Chakraborty, S. Dey, and M. Franceschetti, ``Outage capacity of MIMO Poisson fading channels," \emph{IEEE Trans. Inf. Theory}, vol. 54, no. 11, pp. 4887-4907, Nov. 2008.

\bibitem{BIB08}
A. A. Farid and S. Hranilovic, ``Capacity bounds for wireless optical intensity channels with Gaussian noise," \emph{IEEE Trans. Inf. Theory}, vol. 56, no. 12, pp. 6066-6077, Dec. 2010.

\bibitem{BIB09}
A. A. Farid and S. Hranilovic, ``Channel capacity and non-uniform signaling for free-space optical intensity channels," \emph{IEEE J. Sel. Areas Commun.}, vol. 27, no. 9, pp. 1553-1363, Dec. 2009.

\bibitem{BIB10}
A. Lapidoth, S. M. Moser, and M. A. Wigger, ``On the capacity of free-space optical intensity channels," \emph{IEEE Trans. Inf. Theory}, vol. 55, no. 10, pp. 4449-4461, Oct. 2009.

\bibitem{BIB11}
S. Hranilovic and F. R. Kschischang, ``Capacity bounds for power- and band-limited optical intensity channel corrupted by Gaussian noise," \emph{IEEE Trans. Inf. Theory}, vol. 50, no. 5, pp. 784-795, May 2004.

\bibitem{BIB12}
M. Katz and S. Shamai (Shitz), ``On the capacity-achieving distribution of the discrete-time noncoherent and partially coherent AWGN channels," \emph{IEEE Trans. Inf. Theory}, vol. 50, no. 10, pp. 2257-2270, Oct. 2004.

\bibitem{BIB13}
X. Li, J. Vucic, V. Jungnickel, and J. Armstrong, ``On the capacity of intensity-modulated direct-detection systems and the information rate of ACO-OFDM for indoor optical wireless applications," \emph{IEEE Trans. Commun.}, vol. 60, no. 3, pp. 799-809, Mar. 2012.

\bibitem{BIB14}
R. You and J. M. Kahn, ``Upper-bounding the capacity of optical IM/DD channels with multiple-subcarrier modulation and fixed bias using trigonometric moment space method," \emph{IEEE Trans. Inf. Theory}, vol. 48, no. 2, pp. 514-523, Feb. 2002.

\bibitem{BIB17}
J.-B. Wang, Q.-S. Hu, J. Wang, M. Chen, Y.-H. Huang, and J.-Y. Wang, ``Capacity analysis for dimmable visible light communications," in \emph{IEEE Int. Conf. Commun.}, Sydney, Australia, 2014, pp. 3337-3341.

\bibitem{BIB18}
J.-B. Wang, Q.-S. Hu, J. Wang, M. Chen, and J.-Y. Wang, ``Tight bound on channel capacity for dimmable visible light communications," \emph{J. Lightwave Technol.}, vol. 31, no. 23, pp. 3771-3779, Dec. 2013.

\bibitem{BIB19}
J.-Y. Wang, J.-B. Wang, N. Huang, and M. Chen, ``Capacity analysis for pulse amplitude modulated visible light communications with dimming control," \emph{J. Opt. Soc. Am. A}, vol. 31, no. 3, pp. 561-568, Mar. 2014.

\bibitem{BIB19_1}
X. Bao, X. Zhu, T. Song, and Y. Ou, ``Protocol design and capacity analysis in hybrid network of visible light communication and OFDMA systems," \emph{IEEE Trans. Veh. Technol.}, vol. 63, no. 4, pp. 1770-1778, May 2014.


\bibitem{BIB22}
S. M. Moser, ``Capacity results of an optical intensity channel with input-dependent Gaussian noise," \emph{IEEE Trans. Inf. Theory}, vol. 58, no. 1, pp. 207-223, Jan. 2012.

\bibitem{add00}
M. Soltani and Z. Rezki, ``Optical wiretap channel with input-dependent Gaussian noise under peak- and average-intensity constraints,"\emph{ IEEE Trans. Inf. Theory}, vol. 64, no. 10, pp. 6878-6893, Oct. 2018.

\bibitem{add01}
J.-Y. Wang, Z. Yang, Y. Wang, and M. Chen, ``On the performance of spatial modulation-based optical wireless communications," \emph{IEEE Photon. Technol. Lett.}, vol. 28, no. 19, pp. 2094-2097, Oct. 2016.

\bibitem{add1}
Q. Gao, S. Hu, C. Gong, and Z. Xu, ``Modulation designs for visible light communications with signal-dependent noise," \emph{J. Lightwave Technol.}, vol. 34, no. 23, pp. 5516-5525, Dec. 2016.

\bibitem{add2}
Q. Gao, C. Gong, and Z. Xu, ``Joint transceiver and offset design for visible light communications with input-dependent shot noise," \emph{IEEE Trans. Wirel. Commun.}, vol. 16, no. 5, pp. 2736-2747, May 2017.


\bibitem{BIB23}
J. Katz, ``Detectors for optical communications: A review," \emph{TDA Progress Rep.}, vol. 42, no. 75, pp. 21-38, Nov. 1983.

\bibitem{BIB25}
T. Cover and J. Thomas, \emph{Elements of Information Theory} (2nd ed.), Hoboken: Wiley, 2006.

\bibitem{BIB24}
J. M. Kahn and J. R. Barry, ``Wireless infrared communications," \emph{Proc. IEEE}, vol. 85, no. 2, pp. 265-298, Feb. 1997.


\bibitem{BIB25_1}
M. Reznikov, \emph{Variational Method}, New York: LAP Lambert Academic Publishing, 2012.

\bibitem{BIB26}
S. M. Moser, ``Duality-based bounds on channel capacity," Ph.D. Dissertation, ETH Zurich, Switzerland, Jan. 2005.

\bibitem{BIB27}
I. Csiszar and J. Korner, \emph{Information Theory: Coding Theorems for Discrete Memoryless Systems}, New York: Academic, 1981.

\bibitem{BIB28}
A. Lapidoth and S. M. Moser, ``Capacity bounds via duality with applications to multiple-antenna systems on flat fading channels," \emph{IEEE Trans. Inf. Theory}, vol. 49, no. 10, pp. 2426-2467, Oct. 2003.

\bibitem{BIB29}
A. Lapidoth and S. M. Moser, ``The fading number of single-input multiple-output fading channels with memory," \emph{IEEE Trans. Inf. Theory}, vol. 52, no. 2, pp. 437-453, Feb. 2006.

\end{thebibliography}
\end{document}